\PassOptionsToPackage{colorlinks=true, allcolors=blue}{hyperref}
\documentclass[sigconf]{acmart}
\usepackage[linesnumbered,ruled,vlined]{algorithm2e} 
\usepackage{tikz}
\usepackage{graphicx}
\usepackage{amsthm}
\usepackage{subfigure}
\usepackage{float} 
\usepackage{makecell}
\usepackage{colortbl}
\usepackage{caption}
\usepackage{amsmath}
\usepackage{enumitem}

\newtheorem{problem}{Problem}
\newtheorem{definition}{Definition}

\newtheorem{lemma}{Lemma}
\newtheorem{corollary}{Corollary}

\setcopyright{acmlicensed}
\copyrightyear{2018}
\acmYear{2018}
\acmDOI{XXXXXXX.XXXXXXX}

\acmISBN{978-1-4503-XXXX-X/2018/06}

\begin{document}

\title{Maximum Degree-Based Quasi-Clique Search via an Iterative Framework}

\author{Hongbo Xia}
\affiliation{
    \institution{Harbin Institute of Technology, Shenzhen}
    \city{Shenzhen}
    \country{China}
}
\email{24s151167@stu.hit.edu.cn}

\author{Kaiqiang Yu}
\affiliation{
    \institution{Nanyang Technological University}
    \city{Singapore}
    \country{Singapore}
}
\email{kaiqiang002@e.ntu.edu.sg}

\author{Shengxin Liu}
\affiliation{
    \institution{Harbin Institute of Technology, Shenzhen}
    \city{Shenzhen}
    \country{China}
}
\email{sxliu@hit.edu.cn}
\authornote{Shengxin Liu is the corresponding author.}
    
\author{Cheng Long}
\affiliation{
    \institution{Nanyang Technological University}
    \city{Singapore}
    \country{Singapore}
}
\email{c.long@ntu.edu.sg}

\author{Xun Zhou}
\affiliation{
  \institution{Harbin Institute of Technology, Shenzhen\\Pengcheng Laboratory}
  \city{Shenzhen}
  \country{China}
}
\email{zhouxun2023@hit.edu.cn}

\begin{abstract}
Cohesive subgraph mining is a fundamental problem in graph theory with numerous real-world applications, such as social network analysis and protein-protein interaction modeling. Among various cohesive subgraphs, the $\gamma$-quasi-clique is widely studied for its flexibility in requiring each vertex to connect to at least a $\gamma$ proportion of other vertices in the subgraph. However, solving the maximum $\gamma$-quasi-clique problem is NP-hard and further complicated by the lack of the hereditary property, which makes designing efficient pruning strategies challenging. Existing algorithms, such as \texttt{DDA} and \texttt{FastQC}, either struggle with scalability or exhibit significant performance declines for small values of $\gamma$.
In this paper, we propose a novel algorithm, \texttt{IterQC}, which reformulates the maximum $\gamma$-quasi-clique problem as a series of $k$-plex problems that possess the hereditary property. \texttt{IterQC} introduces a non-trivial iterative framework and incorporates two key optimization techniques: (1) the \emph{pseudo lower bound (pseudo LB)} technique, which leverages information across iterations to improve the efficiency of branch-and-bound searches, and (2) the \emph{preprocessing} technique that reduces problem size and unnecessary iterations. Extensive experiments demonstrate that \texttt{IterQC} achieves up to four orders of magnitude speedup and solves significantly more graph instances compared to state-of-the-art algorithms \texttt{DDA} and \texttt{FastQC}.
\end{abstract}




\maketitle

\section{Introduction}\label{sec-intro}
Cohesive subgraph mining is a fundamental problem in graph theory with numerous real-world applications~\cite{Chang2018cosub, Fang2020cosub, Huang2019cosub, Lee2010cosub}. Unlike cliques, which require complete connectivity, cohesive subgraphs provide more flexibility by allowing the absence of certain edges. In real-world graphs, data is often noisy and incomplete, making fully connected cliques impractical. By allowing a certain degree of edge absence, the relaxation offers a robust way to identify strongly connected substructures while accommodating missing or uncertain relationships. Various approaches have been proposed to relax the constraints of cohesive subgraph definitions, such as $k$-plex~\cite{seidman1978kplex}, $k$-core~\cite{Seidman1983}, and $k$-defective clique~\cite{yu2006defective}.
One widely studied definition is the $\gamma$-quasi-clique~\cite{MIH99}. Given a fraction $\gamma$ between 0 and 1, a $\gamma$-quasi-clique requires that every vertex in the subgraph is directly connected to at least a $\gamma$ proportion of the other vertices in the subgraph. 
Cohesive subgraph mining in the context of $\gamma$-quasi-clique has recently attracted increasing interest~\cite{zeng2007out,LW08,JP09,Pastukhov2018gamma,guo2020scalable,sanei-mehri2021largest,khalil2022parallel,yu2023fast}. It has numerous applications in real-world graph structure analysis, including social network analysis~\cite{Bedi2016CommunityDetection,guo2022directed,Fang2020cosub} and the modeling of protein-protein interaction networks~\cite{Suratanee2014ProteinInteractions,Bader2003,BB09}, which captures relationships between proteins. For example, in~\cite{pei2005crossgraph}, researchers identify biologically significant functional groups by mining large $\gamma$-quasi-cliques that meet a minimum size threshold across various protein-protein and gene-gene interaction networks. The idea is that within a functional protein group, most members interact frequently, suggesting a high likelihood of forming a quasi-clique~\cite{pei2005crossgraph}. Similarly, in another case study~\cite{guo2022directed}, $\gamma$-quasi-cliques are utilized to uncover meaningful communities within networks derived from publication data.

In this paper, we study the \emph{maximum $\gamma$-quasi-clique} problem, which aims to identify the $\gamma$-quasi-clique with the largest number of vertices in a graph. As a natural extension of the classic maximum clique problem, this problem is unsurprisingly NP-hard~\cite{MIH99,Pastukhov2018gamma}.
An even more challenging lies in the fact that the $\gamma$-quasi-clique lacks the hereditary property, i.e., subgraphs of a $\gamma$-quasi-clique are not guaranteed to also be $\gamma$-quasi-cliques. This limitation complicates the design of efficient pruning methods, unlike problems for $k$-plex or $k$-defective clique, where the hereditary property enables effective optimizations.
In this paper, we focus on designing practically efficient exact algorithms to tackle this challenging problem.

\noindent \underline{\textbf{Existing state-of-the-art algorithms.}} 
The state-of-the-art algorithm for the maximum $\gamma$-quasi-clique problem is \texttt{DDA}~\cite{Pastukhov2018gamma}, which combines graph properties with techniques from operations research. A key feature of \texttt{DDA} is its iterative \emph{enumeration (or estimation)} of the degree of the minimum-degree vertex in the maximum $\gamma$-quasi-clique. For each candidate degree, it invokes an integer programming (IP) solver multiple times to solve equivalent sub-problems and identify the final solution. Unlike other approaches relying solely on IP formulations or branch-and-bound techniques based on graph properties, \texttt{DDA} effectively integrates both. However, it has two limitations: (1) it naively enumerates possible minimum degree values, potentially requiring a large number of trials (e.g., \texttt{DDA} may enumerate $O(n)$ values, where $n$ is the number of vertices in the graph); (2) the IP solver operates as a black box, which makes it difficult to optimize for this specific problem.

Another closely related algorithm, \texttt{FastQC}~\cite{yu2023fast}, is the state-of-the-art for enumerating all maximal $\gamma$-quasi-cliques in a graph. \texttt{FastQC} uses a divide-and-conquer strategy with efficient pruning techniques that leverage intrinsic graph properties to systematically and effectively enumerate all maximal $\gamma$-quasi-cliques. It is clear that the largest maximal $\gamma$-quasi-clique identified corresponds to the maximum $\gamma$-quasi-clique. However, since \texttt{FastQC} is not specifically designed to solve the maximum $\gamma$-quasi-clique problem, even with additional pruning methods tailored for the maximum solution, its efficiency remains suboptimal. Moreover, our experimental results in Section~\ref{sec:exp} show that the efficiency of the \texttt{FastQC} algorithm significantly decreases when the value of $\gamma$ is relatively small.

In summary, while both \texttt{DDA} and \texttt{FastQC} have practical strengths, they fail to fully address the efficiency challenges of the maximum $\gamma$-quasi-clique problem. Their limitations in scalability and handling smaller values of $\gamma$ highlight the need for more specialized solutions.

\noindent \underline{\textbf{Our new methods.}}
In this work, we propose a novel algorithm, \texttt{IterQC}, which reformulates the maximum $\gamma$-quasi-clique problem as a series of $k$-plex problems. Notably, the $k$-plex possesses the hereditary property, enabling the design of more effective pruning methods. 
To achieve this, we introduce a non-trivial \emph{iterative framework} that solves the maximum $\gamma$-quasi-clique problem through a carefully designed iterative procedure, leveraging repeated calls to a maximum $k$-plex solver~\cite{wang2023kplex,chang2024maximum,gao2024maximum}.
Building on this novel iterative framework, we further propose advanced optimization techniques, including the \emph{pseudo lower bound (pseudo LB)} technique and the \emph{preprocessing} technique. The pseudo LB technique effectively coordinates information between the inner iterations and the outer iterative framework, thereby boosting the overall efficiency of the iterative procedure. Meanwhile, the preprocessing technique performs preliminary operations before the iterative framework begins, reducing both the problem size by removing redundant vertices from the graph and reducing the number of iterations by skipping unnecessary iterations. With these optimizations, our proposed algorithm \texttt{IterQC} equipped with these optimizations achieves up to four orders of magnitude speedup and solves more instances compared to the state-of-the-art algorithms \texttt{DDA} and \texttt{FastQC}.

\noindent \underline{\textbf{Contributions.}} Our main contributions are as follows.
\begin{itemize}[leftmargin=*]
    \item We first propose a basic iterative framework, which correctly transforms non-hereditary problems into multiple problem instances with the hereditary property. (Section~\ref{sec:basic-iterative})
    \item Based on the basic iterative method, we design an improved algorithm \texttt{IterQC} with two key components: 
    (1) The \emph{pseudo lower bound (pseudo LB)} technique, which utilizes information from previous iterations to generate a pseudo lower bound. This technique optimizes the branch-and-bound search, improving performance on challenging instances. 
    (2) The \emph{preprocessing} technique, which computes lower and upper bounds of the optimum size to remove unpromising vertices from the graph, potentially reducing the number of iterations. (Section~\ref{sec:improved})
    \item We conduct extensive experiments to compare \texttt{IterQC} with the state-of-the-art algorithms \texttt{DDA} and \texttt{FastQC}. The results show that (1) on the 10th DIMACS and real-world graph collections, the number of instances solved by \texttt{IterQC} within 3 seconds exceeds the number solved by the other two baselines within 3 hours; (2) on 30 representative graphs, \texttt{IterQC} is up to four orders of magnitude faster than the state-of-the-art algorithms. (Section~\ref{sec:exp})
\end{itemize}

\section{Preliminaries}\label{sec:prelim}
Let $G=(V, E)$ be an undirected simple graph with $|V| = n$ vertices and $|E|=m$ edges.
We denote $G[S]$ of $G$ as the subgraph induced by the set of vertices $S$ of $V$.
We use $d_G(v)$ to denote the degree of $v$ in $G$.
Let $g$ be an induced subgraph of $G$. We denote the vertex set and the edge set of $g$ as $V(g)$ and $E(g)$, respectively.
In this paper, we focus on the cohesive subgraph of $\gamma$-quasi-clique defined below.
\begin{definition}
Given a graph $G=(V, E)$ and a real number $0 < \gamma \leq 1$, an induced subgraph $g$ is said to be a {\em $\gamma$-quasi-clique ($\gamma$-QC)} if $\forall v \in V(g)$, $d_{g}(v) \geq \gamma \cdot (|V(g)|-1)$.
\end{definition}

We are ready to present our problem in this paper.
\begin{problem}[Maximum $\gamma$-quasi-clique]
Given a graph $G=(V, E)$ and a real number $\gamma \in [0.5, 1]$, the Maximum $\gamma$-quasi-clique Problem (\texttt{MaxQC}) aims to find the largest $\gamma$-quasi-clique in $G$.
\end{problem}
We first note that \texttt{MaxQC} has been proven to be NP-hard~\cite{MIH99,Pastukhov2018gamma}.
Moreover, we let $g^*$ be the largest $\gamma$-QC in $G$ and $s^*=|V(g^*)|$ be the size of this maximum solution.
We also note that, following previous studies~\cite{guo2020scalable,khalil2022parallel,pei2005crossgraph,sanei-mehri2021largest,yu2023fast}, we consider \texttt{MaxQC} with $\gamma \geq 0.5$ since (1) conceptually, a $\gamma$-QC with $\gamma < 0.5$ may not be cohesive in practice (note that each vertex in such a $\gamma$-QC may connect to fewer than half of the other vertices); (2) technically, a $\gamma$-QC with $\gamma \geq 0.5$ has a small diameter of at most 2~\cite{pei2005crossgraph}.

\section{Our Basic Iterative Framework}\label{sec:basic-iterative}
The design of efficient exact algorithms for \texttt{MaxQC} presents multiple challenges. \textbf{First}, \texttt{MaxQC} has been proven to be NP-hard~\cite{MIH99,Pastukhov2018gamma}. Existing studies on similar maximum cohesive subgraph problems often adopt the branch-and-bound algorithmic framework~\cite{chang2022efficient,chang2024maximum,wang2023kplex}.
\textbf{Second}, the $\gamma$-quasi-clique is non-hereditary~\cite{yu2023fast}, i.e., an induced subgraph of a $\gamma$-QC is not necessarily a $\gamma$-QC. 
This non-hereditary property of the $\gamma$-QC complicates the design of effective bounding techniques within the branch-and-bound framework.

To address these challenges, our algorithm adopts an \emph{iterative framework} that solves \texttt{MaxQC} by iteratively invoking a solver for another cohesive subgraph problem.
The design rationale originates from the intention to \emph{convert the non-hereditary subgraph problem into multiple instances of a hereditary subgraph problem, so as to speed up the overall procedure}.

We first introduce another useful cohesive subgraph structure.
\begin{definition}[\cite{seidman1978kplex}]
Given a graph $G=(V, E)$ and an integer $k\geq 1$, an induced subgraph $g$ is said to be a {\em $k$-plex} if $\forall v \in V(g)$, $d_{g}(v) \geq |V(g)|-k$.
\end{definition}

In the literature, a related problem for $k$-plex is defined~\cite{chang2022efficient, Gao2018, jiang2023refined, jiang2021, wang2023kplex, Xiao2017maxkplex, zhou2021improving,Balasundaram2011}, as shown in the following.

\begin{problem}[Maximum $k$-plex]\label{def:plex}
Given a graph $G=(V, E)$ and a positive integer $k$, the Maximum $k$-plex Problem aims to find the largest $k$-plex in $G$. 
\end{problem}

Based on the maximum $k$-plex problem, we next describe our basic iterative framework for the \texttt{MaxQC} problem. 
To simplify the description, we define the following two functions.

\begin{itemize}[leftmargin=*]
        \item \texttt{get-k}$(x) := \lfloor (1 - \gamma) \cdot (x - 1)\rfloor + 1$, which takes a number $x$ of vertices as input and returns a appropriate value of $k$;
        \item \texttt{solve-plex}$(y) := \text{the size of the largest } y\text{-plex in } G$, which takes a value $y$ as input.
\end{itemize}

\begin{algorithm}[t]
    \caption{A Basic Iterative Framework}\label{alg:basic-framework}
    \KwIn{A graph $G=(V,E)$ and a real value $0.5 \leq \gamma \leq  1$}
    \KwOut{The size of the maximum $\gamma$-quasi-clique in $G$}  
    $s_0 \gets |V|$; $k_1 \gets \texttt{get-k}(s_0)$; $i\gets 1$\;
    \While{true}{
        $s_i \gets$ \texttt{solve-plex}($k_i$)\;
        \If{$k_i = \emph{\texttt{get-k}}(s_i)$}{
            $s^* \gets s_i$; \Return{$s^*$}\;
        }
        $k_{i+1} \gets \texttt{get-k}(s_i)$; $i \gets i+1$\;
    }
\end{algorithm}

We present our basic iterative framework for solving \texttt{MaxQC} in Algorithm~\ref{alg:basic-framework}. Specifically, Algorithm~\ref{alg:basic-framework} initializes $s_0$ as the number of vertices in the graph $G$, computes the corresponding value of $k_1$ using the function \texttt{get-k}, and sets the index $i$ as 1 (Line 1). The algorithm then iteratively computes $s_i$ by solving the maximum $k_i$-plex problem via the function \texttt{solve-plex}, while updating $k_i$ via the function \texttt{get-k} in each iteration (Lines 2-6). The iteration terminates when $k_i = \texttt{get-k}(s_i)$ in Line 4, at which point $s^*$, the size of the maximum $\gamma$-quasi-clique, is returned in Line 5.
Note that Algorithm~\ref{alg:basic-framework} returns only the size of the maximum $\gamma$-quasi-clique, but can be easily modified to output the corresponding subgraph.

\noindent \underline{\textbf{Correctness.}}
We now prove the correctness of our basic iterative framework in Algorithm~\ref{alg:basic-framework}.
For simplicity in the proof, we assume that the termination condition (Line 4) of Algorithm~\ref{alg:basic-framework} is met at the $(p+1)$-st iteration. Then, in our iterative framework, we can obtain two sequences $\{s_0,s_1,\ldots,s_{p},s_{p+1}\}$ and $\{k_1,k_2,\ldots,k_{p},k_{p+1}\}$, where $k_{p+1}=\texttt{get-k}(s_{p+1})$.
In our proof, we will show that the sequence of $\{s_0,s_1,\ldots,s_{p}\}$ generated by Algorithm~\ref{alg:basic-framework} is strictly decreasing and $s_{p+1}$ corresponds to the largest $\gamma$-quasi-clique in the graph, i.e, $s_{p+1} = s^*$.
We first consider the special case that the input graph $G$ is already a $\gamma$-quasi-clique.
\begin{lemma}\label{lem:input-qc}
    When the input graph $G=(V,E)$ is a $\gamma$-QC, Algorithm~\ref{alg:basic-framework} correctly returns $|V|$ as the optimum solution.
\end{lemma}
The proof of Lemma~\ref{lem:input-qc}, along with other omitted proofs in this section, can be found in Section~\ref{sec:proof-basic}.
In the following, we discuss the correctness of the algorithm where $G$ itself is not a $\gamma$-quasi-clique.
We first prove the properties of the sequence for $s$ in the following.

\begin{lemma}\label{lem:seq-s-non-identical}
    For the sequence $\{s_0,s_1,\ldots,s_{p}\}$, two consecutive elements 
    are not identical, i.e., $\forall 0 \leq i \leq p-1$, $s_i \neq s_{i+1}$.
\end{lemma}

\begin{lemma}\label{lem:seq-s}
The sequence $\{s_0,s_1,\ldots,s_p\}$ is strictly decreasing.
\end{lemma}
\begin{proof}
    We prove this lemma by the mathematical induction.\\
    \textcircled{\scriptsize{1}}\normalsize\enspace $p=1$. As we consider the case where $G$ itself is not a $\gamma$-quasi-clique (by Lemma~\ref{lem:input-qc}), we have $s_p=s_1=\texttt{solve-plex}(k_1) < s_0$. The base case holds true. \\
    \textcircled{\scriptsize{2}}\normalsize\enspace $p \geq 2$. Assume that the induction holds for $p-1$, i.e., we have $s_{p-1} < s_{p-2}$.
    We observe that \texttt{get-k} is a non-decreasing function, which implies that $\texttt{get-k}(s_{p-1}) \leq  \texttt{get-k}(s_{p-2})$. 
    Moreover, by the definition of the \texttt{solve-plex} function, for any $y_1 > y_2$, the inequality $\texttt{solve-plex}(y_1) \geq \texttt{solve-plex}(y_2)$ holds. Then we have
\begin{align*}
s_p &= \texttt{solve-plex}(k_p) = \texttt{solve-plex}(\texttt{get-k}(s_{p-1}))\\
&\leq \texttt{solve-plex}(\texttt{get-k}(s_{p-2}))=\texttt{solve-plex}(k_{p-1})=s_{p-1}.
\end{align*}
Note that based on Lemma~\ref{lem:seq-s-non-identical} the sequence $\{s_0,s_1,\ldots,s_p\}$ ensures that no two adjacent elements share the same value, which implies $s_p < s_{p-1}$, completing the inductive step.

Based on \textcircled{\scriptsize{1}}, \textcircled{\scriptsize{2}}, and the principle of mathematical induction, we complete the proof of Lemma~\ref{lem:seq-s}.
\end{proof}

Based on Lemma~\ref{lem:seq-s}, we can obtain the following corollary.

\begin{corollary}\label{cor:sequence}
    The sequence $\{s_0,s_1,\ldots,s_p\}$ is finite.
\end{corollary}

We then show the relationship between $\{s_0,s_1,\ldots,s_p,s_{p+1}\}$ and the size of the largest $\gamma$-quasi-clique.
\begin{lemma}\label{lem:seq-s-least}
$\forall s_i \in \{s_0,s_1,\ldots,s_{p+1}\}, \, s_i \geq s^*$ holds.
\end{lemma}

\begin{lemma}\label{lem:correct-ans}
    Algorithm~\ref{alg:basic-framework} correctly identifies the largest $\gamma$-quasi-clique, i.e., $s_{p+1} = s^*$.
\end{lemma}
\begin{proof}
    We prove by contradiction. By Lemma~\ref{lem:seq-s-least}, we assume, to the contrary, that $s_{p+1} > s^*$. According to the termination condition in Line 4 of Algorithm~\ref{alg:basic-framework}, we have $k_{p+1} = \texttt{get-k}(s_{p+1}) = 1 + \lfloor (1-\gamma) \cdot (s_{p+1} - 1)\rfloor$. Thus, it is clear that the result obtained by $\texttt{solve-plex}(k_{p+1})$ is a $\gamma$-quasi-clique of size $s_{p+1}$, which contradicts $s^*$ being the size of the maximum $\gamma$-quasi-clique.
\end{proof}

\noindent \underline{\textbf{Discussions.}}
We first note that an example of Algorithm~\ref{alg:basic-framework} and its iterative process are provided in Appendix~\ref{sec:example}.
As shown in Lemma~\ref{lem:correct-ans}, Algorithm~\ref{alg:basic-framework} can correctly solve \texttt{MaxQC}. The time complexity analysis is simple, since there are at most $n$ iterations (Lines 2-6) and each iteration invokes a solver for the maximum $k$-plex problem. Note that the current best time complexities of the maximum $k$-plex algorithms are $O^*((k+1)^{\delta + k - s^*} )$ and $O^*(\gamma_k ^ \delta)$~\cite{wang2023kplex,chang2024maximum,gao2024maximum}, where ${O^*}$ suppresses the polynomial factors, $\gamma_k < 2$ is the largest real root of $x^{k+3} -2x^{k+2} + x^2 - x + 1 = 0$, and $\delta$ is the degeneracy of the graph. 
We also remark that prior studies adopted similar approaches by reducing a more difficult problem to solving another problem iteratively. For instance, Chang and Yao~\cite{chang2024maximum} briefly discuss the relationship between the $\gamma$-quasi-clique and the $k$-plex. They focus on enumerating maximal $\gamma$-quasi-cliques and mention a method for generating an initial solution set by enumerating maximal $k$-plexes, followed by a screening step to remove those that do not satisfy the maximality condition. However, this method is mainly of theoretical interest and runs slowly in practice, since it requires enumerating a large number of maximal $k$-plexes. Zhang and Liu~\cite{zhang2023mincore} tackle the minimum $k$-core problem by reducing it to multiple iterations of solving the maximum $k$-plex problem. Their iterative strategy incrementally adjusts the value of $k$ by one in each round, progressively approaching the desired solution. Similarly, they also adopt an iterative strategy that considers different values of $k$, where the value of $k$ differs by one between consecutive iterations.
In contrast, we consider the maximum $\gamma$-quasi-clique problem. Our method adaptively adjusts the values of $k$ based on the size of the current best $k$-plex, rather than exhaustively enumerating all possible values of $k$.

However, Algorithm~\ref{alg:basic-framework} still suffers from efficiency issues for the following reasons.
{\bf First}, Algorithm~\ref{alg:basic-framework} utilizes a solver for \texttt{solve-plex} as a black box. This solver relies on a conservative lower bound within \texttt{solve-plex} to reduce the graph, leading to inefficiencies when handling certain challenging instances.
{\bf Second}, the graph processed iteratively in \texttt{solve-plex} (Line 3) often contains many unpromising vertices, which negatively affects overall performance, even after graph reduction within \texttt{solve-plex}. Further, in the initial iteration, the first value of $k$ is set as $k_1=\texttt{get-k}(|V|)$ with $|V|$ is a trivial upper bound of $s^*$ in Line 1. This initialization could potentially result in numerous unnecessary iterations of Lines 2-6.

\section{Our Improved Algorithm: \texttt{IterQC}}\label{sec:improved}
To address the limitations of our basic iterative framework (Algorithm~\ref{alg:basic-framework}), we propose an advanced iterative method \texttt{IterQC} (shown in Algorithm~\ref{alg:improved-framwork}), which incorporates two key stages: the preprocessing stage and the iterative search stage. 
\textbf{First}, we consider the iterative search stage (Line 4). Within the function \texttt{Plex-Search} -- an improved version of \texttt{solve-plex} from Algorithm~\ref{alg:basic-framework} -- we introduce a \emph{pseudo lower bound (pseudo LB)}. This technique improves practical performance when handling challenging instances while preserving the correctness of the iterative algorithm for \texttt{MaxQC}. The pseudo LB technique is discussed in Section~\ref{subsec-pseudolb}.
\textbf{Second}, we provide a \emph{preprocessing} technique (Lines 1-3) to improve efficiency. Specifically, we compute both the lower and upper bounds, \( lb \) and \( ub \), of the size \( s^* \) of the largest \( \gamma \)-quasi-clique \( g^* \) via \texttt{Get-Bounds} in Line 1. Using \(lb\) and \(ub\), in Lines 2-3, we can (1) reduce the graph by removing unpromising vertices/edges from $G$ that cannot appear in $g^*$, and (2) initialize a smaller value of $k$ for \texttt{Plex-Search}, potentially reducing the number of iterations. The preprocessing technique is introduced in Section~\ref{subsec-prepro}.
Finally, we analyze the time complexity of our \texttt{IterQC}, which theoretically improves upon the state-of-the-art algorithms, in Section~\ref{sec:improved-time}

\begin{algorithm}[t]
    \caption{Our Improved Algorithm: \texttt{IterQC}}\label{alg:improved-framwork}
    \KwIn{A graph $G=(V,E)$ and a real value $0.5 \leq \gamma \leq 1$}
    \KwOut{The size of the maximum $\gamma$-quasi-clique in $G$} 
    \tcp{Stage 1: Preprocessing}
    $lb, ub \gets \texttt{Get-Bounds}(G, \gamma)$ \;

    \While{$\min_{v \in V} d_G(v) < \lfloor(lb - 1) \cdot \gamma \rfloor$}{
        $u \gets \arg\min_{v \in V} d_G(v)$; Remove $u$ from $G$\;
    }    
    \tcp{Stage 2: Iterative Search}
    \texttt{Improved-Iter-Search$(G, \gamma, ub)$}\;
\end{algorithm}

\subsection{A Novel Pseudo Lower Bound Technique}\label{subsec-pseudolb}
The basic iterative framework (Algorithm~\ref{alg:basic-framework}) solves \texttt{MaxQC} by iteratively invoking the maximum $k$-plex solver \texttt{solve-plex} as a black box with varying values of $k$. To improve practical efficiency, we have the following key observation.

\noindent \underline{\textbf{Key observation.}} Existing studies~\cite{Xiao2017maxkplex,chang2022efficient,jiang2023refined,wang2023kplex,chang2024maximum,gao2024maximum} on maximum $k$-plex solvers mainly consist of two steps: (1) heuristically computing a lower bound of the maximum $k$-plex, and (2) exhaustively conducting a branch-and-bound search to find the final result.
It is well known that Step (1) has a polynomial complexity while Step (2) requires an exponential complexity. Thus, both theoretically and practically, in most cases, the time cost of the branch-and-bound search dominates that of obtaining the lower bound. 

With the above property, our key idea is to \emph{balance} the time cost of both steps. In particular, we introduce a \emph{pseudo lower bound}, denoted by $pseudo\mbox{-}lb$, which is defined as the average of the heuristic lower bound from Step (1) and a known upper bound for the maximum $k$-plex solution. In other words, $pseudo\mbox{-}lb$ is always at least the heuristic lower bound and at most the known upper bound.
By incorporating $pseudo\mbox{-}lb$ into the branch-and-bound search in Step (2), we may improve pruning effectiveness, which reduces the overall search time as verified by our experiments. We note that in maximum $k$-plex algorithms, Step (1) corresponds to \texttt{Plex-Heu} for computing a heuristic solution, while Step (2) corresponds to \texttt{Plex-BRB} for performing the branch-and-bound search.


\noindent \underline{\textbf{The technique of pseudo LB.}} We call this technique as \emph{pseudo lower bound (pseudo LB)}, which is incorporated into a branch-and-bound search algorithm called \texttt{Plex-Search} in Algorithm~\ref{alg:kplex-search}.
In \texttt{Plex-Search}, $ub\mbox{-}plex$ represents an upper bound of the size of the maximum $k$-plex found in \texttt{Plex-Search}, meaning no $k$-plex larger than $ub\mbox{-}plex$ exists. The output of \texttt{Plex-Search} include a pseudo lower bound $pseudo\mbox{-}lb$ and a pseudo size $|S|$ which corresponds to (1) the size of the maximum $k$-plex of size at least $pseudo\mbox{-}lb$ if such a $k$-plex exists, or (2) 0, if no $k$-plex of size at least $pseudo\mbox{-}lb$ exists. This is because, when the input $pseudo\mbox{-}lb$ exceeds the size of the maximum $k$-plex in the graph, \texttt{Plex-BRB} will utilize this $pseudo\mbox{-}lb$ for pruning, resulting in an empty graph as the vertex set $S$.


\noindent \underline{\textbf{Our \texttt{Plex-Search} method.}} We describe \texttt{Plex-Search} in Algorithm~\ref{alg:kplex-search}.
Specifically, in Line 1, we invoke \texttt{Plex-Heu} to compute a lower bound $lb\mbox{-}plex$. If $lb\mbox{-}plex = ub\mbox{-}plex$, we can directly return $lb\mbox{-}plex$ as both the pseudo lower bound and the pseudo size in Line 2, since \texttt{Plex-Heu} already finds the maximum $k$-plex in this case. Then, Line 3 computes our pseudo lower bound $pseudo\mbox{-}lb$ as the average of $lb\mbox{-}plex$ and $ub\mbox{-}plex$. We conduct \texttt{Plex-BRB} using $pseudo\mbox{-}lb$ for pruning and obtain the corresponding $k$-plex $S$ (either the vertex set of maximum $k$-plex or an empty set). Finally, we return $pseudo\mbox{-}lb$ and $|S|$ in Line 5.

\begin{algorithm}[t]
    \caption{\texttt{Plex-Search$(G,k,ub\mbox{-}plex)$}}\label{alg:kplex-search}
    \KwOut{a pseudo lower bound $pseudo\mbox{-}lb$ and a pseudo size}
    $lb\mbox{-}plex \gets \texttt{Plex-Heu}(G)$; \tcp{$G$ is locally reduced.}
    \lIf{$lb\mbox{-}plex = ub\mbox{-}plex$}{
        \Return{$lb\mbox{-}plex, lb\mbox{-}plex$}}
    $pseudo\mbox{-}lb \gets \lfloor(lb\mbox{-}plex + ub\mbox{-}plex)/2\rfloor$\;
    $S \gets \texttt{Plex-BRB}(G, pseudo\mbox{-}lb)$; \tcp{The maximum $k$-plex of size at least $pseudo\mbox{-}lb$ if exists; empty, otherwise.}
    \Return{$pseudo\mbox{-}lb, |S|$}\;
\end{algorithm}


\noindent \underline{\textbf{Our improved iterative search algorithm.}} With our newly proposed \texttt{Plex-Search} with the pseudo LB technique, we next present our improved iterative search method \texttt{Improved-Iter-Search} in Algorithm~\ref{alg:improved-iterative-search}. This method is similar to our basic iterative framework in Algorithm~\ref{alg:basic-framework}.
Specifically, Algorithm~\ref{alg:improved-iterative-search} initializes $s_0$ to $ub$, computes the corresponding value of $k_1$ using the function \texttt{get-k}, and sets the index $i$ to 1 (Line 1). The algorithm then iteratively computes $s_i$ by invoking \texttt{Plex-Search} (instead of \texttt{solve-plex} in Algorithm~\ref{alg:basic-framework}), while updating $k_i$ using the function \texttt{get-k} at each iteration (Lines 2-6). Note that both $pseudo\mbox{-}lb$ and $pseudo\mbox{-}size$ are obtained from our \texttt{Plex-Search}, and $s_i$ is set to the greater of these two values in Line 3.
The iteration terminates when $k_i = \texttt{get-k}(s_i)$ and $pseudo\mbox{-}size \geq pseudo\mbox{-}lb$ in Line 4. This termination condition ensures that (1) the output from \texttt{Plex-Search} is the maximum $k$-plex even with the use of a pseudo lower bound; and (2) the current iteration produces the same value of $k$ for the next iteration from \texttt{Plex-Search}. At this point, we return $s^*$ as the size of the maximum $\gamma$-quasi-clique in Line 5. 

\begin{algorithm}[t]
    \caption{\texttt{Improved-Iter-Search$(G, \gamma, ub)$}}
    \label{alg:improved-iterative-search}
    \KwOut{The size of the maximum $\gamma$-quasi-clique in $G$}  
    $s_0 \gets ub$; $k_1 \gets \texttt{get-k}(s_0)$; $i \gets 1$\;
    \While{true}{
        $pseudo\mbox{-}lb, pseudo\mbox{-}size \gets$ \texttt{Plex-Search}($G$, $k_i$, $s_{i-1}$); $s_i \gets \max \{ pseudo\mbox{-}lb, pseudo\mbox{-}size\}$\;
        \If{$k_i = \emph{\texttt{get-k}}(s_i)$ \textbf{\em and} $pseudo\mbox{-}size \geq pseudo\mbox{-}lb$}{
            $s^*\gets s_i$; \Return{$s^*$}\;
        }
        $k_{i+1} \gets \texttt{get-k}(s_i)$; $i \gets i + 1$\;
    }
\end{algorithm}

\noindent \underline{\textbf{Remark.}}
We note that our improved iterative search method in Algorithm~\ref{alg:improved-iterative-search} can use adaptions of any existing maximum $k$-plex algorithm~\cite{Xiao2017maxkplex,chang2022efficient,jiang2023refined,wang2023kplex,chang2024maximum,gao2024maximum} for \texttt{Plex-Search}, provided the algorithm can be decomposed into two steps: (1) heuristically obtaining a lower bound for the maximum $k$-plex, \texttt{Plex-Heu}, and (2) exhaustively conducting the exact branch-and-bound search for the final solution, \texttt{Plex-BRB}. In this work, we utilize one of the state-of-the-art maximum $k$-plex algorithms \texttt{KPEX}~\cite{gao2024maximum}. 

\noindent \underline{\textbf{Correctness.}} We prove the correctness of our improved iterative search algorithm \texttt{Improved-Iterative-Search} in Algorithm~\ref{alg:improved-iterative-search}, which includes the pseudo LB technique in Algorithm~\ref{alg:kplex-search}. 
The complete proof and the illustration example of the pseudo LB technique are provided in Appendix~\ref{sec:proof-improved} and Appendix~\ref{sec:example}, respectively.

\subsection{Our Preprocessing Technique}\label{subsec-prepro}
To show our preprocessing technique, we first describe the computation of lower and upper bounds \texttt{Get-Bounds} (Line 1 of Algorithm~\ref{alg:improved-framwork}).

\noindent \underline{\textbf{Computation of lower and upper bounds.}}  
Following the idea commonly adopted for computing lower and upper bounds in prior dense subgraph search studies~\cite{yu2023fast,chang2022efficient,Chang2023,rahman2024pseudo}, we derive our bounds based on the $k$-core structure. In the following, we introduce the relevant concepts and present the bounds for our problem.
\begin{definition} (\(k\)-core~\cite{Seidman1983})
    Given a graph \(G\) and an integer \(k\), the \(k\)-core of \(G\) is the maximal subgraph \(g\) of \(G\) such that every vertex \(u \in V(g)\) has degree \(d_g(u) \geq k\) in \(g\).
\end{definition}
Based on the $k$-core, a related concept is the \emph{core number} of a vertex \( u \), denoted as \({core}(u)\), which is the largest \( k \) such that \( u \) is part of a \( k \)-core in the graph. In other words, \( u \) cannot belong to any subgraph where all vertices have a degree of at least \(core(u) + 1 \). 
Let $max\_core$ be the maximum core number in the graph, i.e., $max\_core = \max_{u \in V} core(u)$. Given a graph $G=(V,E)$, the core numbers for all vertices can be computed in \( O(|V| + |E|) \) time using the famous peeling algorithm~\cite{Matula83,BVZ03m}, which iteratively removes vertices with the smallest degree from the current graph. 
We remark that during the peeling algorithm used to compute the core numbers, we can obtain the lower bound $lb$ of $s^*$ as a by-product. In particular, we check whether the current graph qualifies as a $\gamma$-quasi clique and record the size of the largest qualifying $\gamma$-quasi clique. Moreover, as the current graph keeps shrinking during the peeling algorithm, the lower bound $lb$ clearly corresponds to the size of the first current graph that meets the $\gamma$-quasi clique requirements.

\begin{lemma}\label{lem:ub}
    Given a graph $G$ and a real number $0 < \gamma \leq 1$, we have $s^* \leq 1 + \lceil max\_core / \gamma \rceil$.
\end{lemma}

The proof of Lemma~\ref{lem:ub} directly follows from the fact that the minimum degree of a $\gamma$-QC is at most $max\_core$ and the definition of $\gamma$-QC.
\begin{algorithm}[t]
    \caption{\texttt{Get-Bounds$(G,\gamma)$}}\label{alg:get-quasi-clique-bound}
    \KwOut{The lower and upper bounds $lb$ and $ub$ of $s^*$}
    
    $lb \gets 0$; $ub \gets 0$; $max\_core \gets 0$\;

    \While{$V(G) \neq \emptyset$}{
        $u \gets \arg\min_{v \in V(G)} d_G(v)$\;
        $max\_core \gets \max(max\_core, d_G(u))$\;
        $ub \gets \max (ub, \min(1 + \lceil max\_core/\gamma \rceil, |V(G)|))$\;
        \If{$G$ is a $\gamma$-quasi-clique \textbf{\em and} $lb = 0$}{
            $lb \gets |V(G)|$\;
        }
        Remove $u$ from $G$\;
    }
    \Return{$lb$, $ub$}\;
\end{algorithm}
Our method \texttt{Get-Bounds} for computing both lower and upper bounds is summarized in Algorithm~\ref{alg:get-quasi-clique-bound}. In particular, as mentioned, Algorithm~\ref{alg:get-quasi-clique-bound} follows the peeling algorithm, which iteratively removes the vertex with the smallest degree from the current graph while dynamically updating $ub$ and $lb$.

\noindent \underline{\textbf{The preprocessing stage.}}  
After obtaining both lower and upper bounds via \texttt{Get-Bounds}, we show how to utilize these bounds in our preprocessing. Recall that, with $lb$ and $ub$, we can (1) reduce the graph by removing unpromising vertices from $G$ that cannot appear in $g^*$ using $lb$, and (2) set a smaller initial value of $k$ in the function \texttt{Plex-Search} using $ub$. The details are as follows.
For (1), we can iteratively remove vertices from \( G \) with degrees less than \( \lfloor(lb - 1) \cdot \gamma \rfloor \) and their incident edges, since it is clear that no solution of size greater than \( lb \) can still include these vertices.
For (2), we can easily derive from Lemma~\ref{cor:iteration with ub} (Section~\ref{alg:improved-iterative-search} where the proofs are in Appendix~\ref{sec:proof-improved}) that, instead of initializing $s_0$ as $|V|$, $s_0$ can be set to an integer larger than $s^*$ while still ensuring our iterative framework produces the correct solution. This allows us to use the obtained upper bound \( ub \) to update \( s_0 \). Additionally, this $ub$ is used to generate the initial value of $k$ for \texttt{Plex-Search} through \texttt{get-k}, while maintaining the correctness of the iterative framework.

\subsection{Time Complexity Analysis}~\label{sec:improved-time}
As mentioned, in our implementation of \texttt{Improved-Iter-Search} (or specifically, \texttt{Plex-Solve} in Algorithm~\ref{alg:kplex-search}), we use two components, i.e., the heuristic and the branch-and-bound search, of the maximum $k$-plex solver \texttt{kPEX}~\cite{gao2024maximum}.
The time complexity of our exact algorithm \texttt{IterQC} is dominated by \texttt{Plex-BRB} in Algorithm~\ref{alg:kplex-search}, which is invoked at most $n$ times. The time complexities of \texttt{kPEX} are $O^*((k+1)^{\delta + k - s^*} )$ and $O^*(\gamma_k ^ \delta)$, where ${O^*}$ suppresses the polynomials, $\gamma_k < 2$ is the largest real root of $x^{k+3} -2x^{k+2} + x^2 - x + 1 = 0$, and $\delta$ is the degeneracy of the graph.
These complexities are derived using the branching methods in~\cite{wang2023kplex,chang2024maximum}, which represent the current best time complexities for the maximum $k$-plex problem.
In addition, as the sequence $\{s_1, s_2, \ldots, s_p\}$ generated by our improved iterative framework is monotonically decreasing and \texttt{get-k} is a non-decreasing function, each value of $k_i$ is thus upper bounded by $k_1$. Thus, the time complexity of each iteration depends on solving the maximum $k$-plex problem with $k_1$. Further, with the preprocessing technique, $k_1 = \delta \frac{1-\gamma}{\gamma}$ and we let $k = k_1$. Therefore, the time complexity of \texttt{IterQC} is given by 
$
O^*\left(\min \left\{ \left(\delta + 1 \right)^{2\delta - s^*}, \gamma_{k}^\delta \right\}\right).
$
This improves upon (1) the time complexity of \texttt{FastQC}~\cite{yu2023fast}, which is $O^*(\sigma_{k}^{\delta \cdot \Delta})$, where $\Delta$ denotes the maximum degree of vertices in the graph and $\gamma_k < \sigma_k$ for each $k$, and improves upon (2) the time complexity of \texttt{DDA}~\cite{Pastukhov2018gamma}, which is $O^*(2^n)$.

\section{Experiments}\label{sec:exp}
We conduct extensive experiments to evaluate the practical performance of our exact algorithm \texttt{IterQC} against other exact methods.

\begin{itemize}[leftmargin=*]
    \item {\texttt{DDA}\footnote{We re-implemented the algorithm based on the description in \cite{Pastukhov2018gamma} using the ILP solver IBM ILOG CPLEX Optimizer (version 22.1.0.0). Our implementation has better performances compared with the results shown in their paper.}:} the state-of-the-art algorithm~\cite{Pastukhov2018gamma}.  
    \item {\texttt{FastQC}\footnote{https://github.com/KaiqiangYu/SIGMOD24-MQCE}:} a baseline adapted from the state-of-the-art algorithm for enumerating all maximal $\gamma$-QCs of size at least a given lower bound~\cite{yu2023fast}. We adapt \texttt{FastQC} for our \texttt{MaxQC} problem by tracking of the largest maximal $\gamma$-QC seen so far and dynamically updating the lower bound to facilitate better pruning. Note that the largest $\gamma$-QC is the largest one among all maximal $\gamma$-QCs. 
\end{itemize} 

\noindent \underline{{\bf Setup.}}  All algorithms are implemented in C++, compiled with g++ \mbox{-O3}, and run on a machine with an Intel CPU @ 2.60GHz and 256GB main memory running Ubuntu 20.04.6. We set the time limit as 3 hours (i.e., 10,800s) and use \textbf{OOT} (Out Of Time limit) to represent the time exceeds the limit.
Our source code can be found at https://github.com/SmartProbiotics/IterQC.

\noindent \underline{{\bf Datasets.}}  We evaluate the algorithms on two graph collections (223 graphs in total) that are widely used in previous studies.
\begin{itemize}[leftmargin=*]
    \item {The \textbf{real-world} collection\footnote{http://lcs.ios.ac.cn/~caisw/Resource/realworld\%20graphs.tar.gz}} contains 139 real-world graphs from Network Repository with up to $5.87 \times 10^7$ vertices.  
    \item {The \textbf{10th DIMACS} collection\footnote{https://sites.cc.gatech.edu/dimacs10/downloads.shtml}} contains 84 graphs from the 10th DIMACS implementation challenge  with up to $5.09 \times {10^7}$ vertices.  
\end{itemize}

To provide a detailed comparison, we select 30 representative graphs as in Table~\ref{table:representive_instances} (where density is computed as $\frac{2m}{n(n-1)}$). 
Specifically, these representative graphs include 10 small graphs (i.e., G1 to G10) with $n <10^5$, 10 medium graphs (i.e., G11 to G20) with $10^5 \leq n < 10^7$, and 10 large graphs (i.e., G21 to G30) with $n \geq 10^7$.

\begin{table}[t]
    \centering
    \captionsetup{font=small}
    \caption{Statistics of 30 representative graphs.}
    \vspace{-0.15in}
    \label{table:representive_instances}
    \scalebox{0.75}{
    \begin{tabular}{c|l|c|c|c|c}
        \hline
        ID & Graph & $n$ & $m$ & density & $\delta$ \\ \hline
        G1 & tech-WHOIS & 7476 & 56943 & $2.04 \cdot 10^{-3}$ & 88 \\ \hline
        G2 & scc\_retweet & 17151 & 24015 & $1.63 \cdot 10^{-4}$ & 19 \\ \hline
        G3 & socfb-Berkeley & 22900 & 852419 & $3.25 \cdot 10^{-3}$ & 64\\ \hline
        G4 & ia-email-EU & 32430 & 54397 & $1.03 \cdot 10^{-4}$ & 22 \\ \hline
        G5 & socfb-Penn94 & 41536 & 1362220 & $1.58 \cdot 10^{-3}$ & 62\\ \hline
        G6 & sc-nasasrb & 54870 & 1311227 & $8.71 \cdot 10^{-4}$ &35\\ \hline
        G7 & socfb-OR & 63392 & 816886 & $4.07 \cdot 10^{-4}$ &52\\ \hline
        G8 & soc-Epinions1 & 75879 & 405740 & $1.41 \cdot 10^{-4}$ &67\\ \hline
        G9 & rec-amazon & 91813 & 125704 & $2.98 \cdot 10^{-5}$ &4 \\ \hline
        G10 & ia-wiki-Talk & 92117 & 360767 & $8.50 \cdot 10^{-5}$ &58\\ \hline
        G11 & soc-LiveMocha & 104103 & 2193083 & $4.05 \cdot 10^{-4}$ &92\\ \hline
        G12 & soc-gowalla & 196591 & 950327 & $4.92 \cdot 10^{-5}$ &51\\ \hline
        G13 & delaunay\_n18 & 262144 & 786395 & $2.29 \cdot 10^{-5}$ &4\\ \hline
        G14 & auto & 448695 & 3314610 & $3.29 \cdot 10^{-5}$ &9\\ \hline
        G15 & soc-digg & 770799 & 5907132 & $1.99 \cdot 10^{-5}$ &236\\ \hline
        G16 & ca-hollywood-2009 & 1069126 & 56306653 & $9.85 \cdot 10^{-5}$ &2208\\ \hline
        G17 & inf-belgium\_osm & 1441295 & 1549969 & $1.49 \cdot 10^{-6}$ &3\\ \hline
        G18 & soc-orkut & 2997166 & 106349209 & $2.37 \cdot 10^{-5}$ &230\\ \hline
        G19 & rgg\_n\_2\_22\_s0 & 4194301 & 30359197 & $3.45 \cdot 10^{-6}$ &19\\ \hline
        G20 & inf-great-britain\_osm & 7733821 & 8156516 & $2.73 \cdot 10^{-7}$ &3\\ \hline
        G21 & inf-asia\_osm & 11950756 & 12711602 & $1.78 \cdot 10^{-7}$ &3\\ \hline
        G22 & hugetrace-00010 & 12057441 & 18082178 & $2.49 \cdot 10^{-7}$ &2\\ \hline
        G23 & inf-road\_central & 14081816 & 16933412 & $1.71 \cdot 10^{-7}$ &3 \\ \hline
        G24 & hugetrace-00020 & 16002413 & 23998812 & $1.87 \cdot 10^{-7}$ &2\\ \hline
        G25 & rgg\_n\_2\_24\_s0 & 16777215 & 132557199 & $9.42 \cdot 10^{-7}$ &20\\ \hline
        G26 & delaunay\_n24 & 16777216 & 50331600 & $3.58 \cdot 10^{-7}$ &4\\ \hline
        G27 & hugebubbles-00020 & 21198119 & 31790178 & $1.41 \cdot 10^{-7}$ &2\\ \hline
        G28 & inf-road-usa & 23947347 & 28854312 & $1.01 \cdot 10^{-7}$ &3\\ \hline
        G29 & inf-europe\_osm & 50912018 & 54054659 & $4.17 \cdot 10^{-8}$ &3\\ \hline
        G30 & socfb-uci-uni & 58790782 & 92208195 & $5.34 \cdot 10^{-8}$ &16\\ \hline
    \end{tabular}}
\end{table}

\subsection{Comparison with Baselines}
\begin{figure}[t]
    \centering
    \subfigure[10th DIMACS (3-hour limit)]{
        \includegraphics[width=0.21\textwidth]{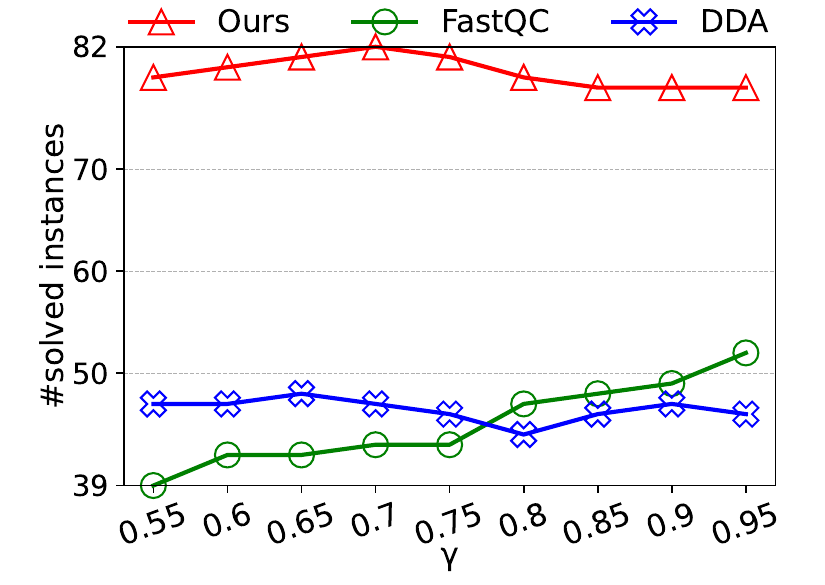}
        \label{fig:dimacs3h}
    }
    \hspace{0em}
    \subfigure[real-world (3-hour limit)]{
        \includegraphics[width=0.21\textwidth]{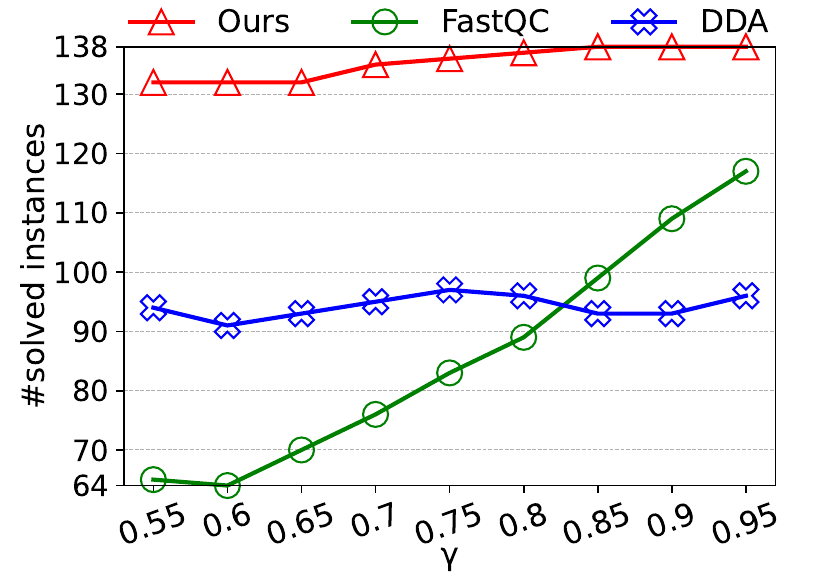}
        \label{fig:realworld3h}
    }
    
    \vspace{0em} 
    
    \subfigure[10th DIMACS (3-second limit)]{
        \includegraphics[width=0.21\textwidth]{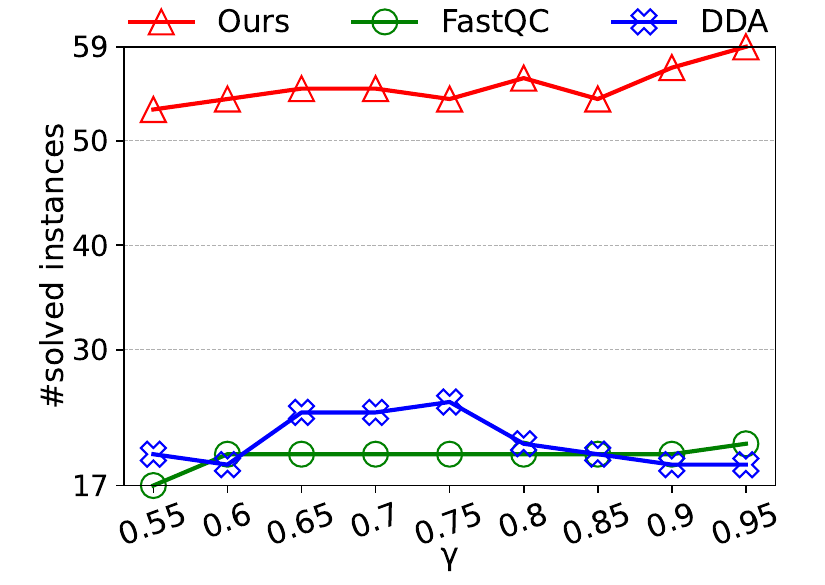}
        \label{fig:dimacs3s}
    }
    \hspace{0em}
    \subfigure[real-world (3-second limit)]{
        \includegraphics[width=0.21\textwidth]{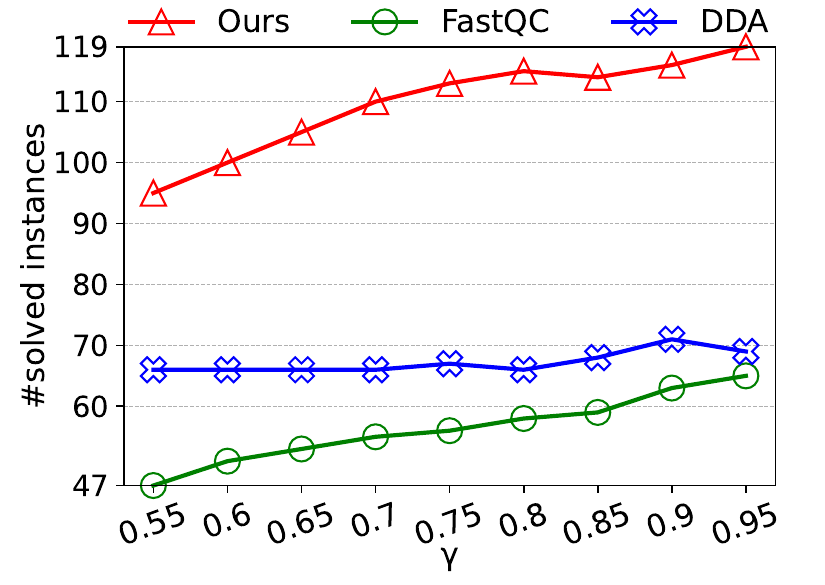}
        \label{fig:realworld3s}
    }
    \vspace{-0.2in}
    \caption{Number of solved instances with varying $\gamma$.}
    \label{Fig:3h}
\end{figure}

\begin{figure}[t]
    \centering
    \subfigure[$\gamma = 0.65$]{
        \includegraphics[width=0.21\textwidth]{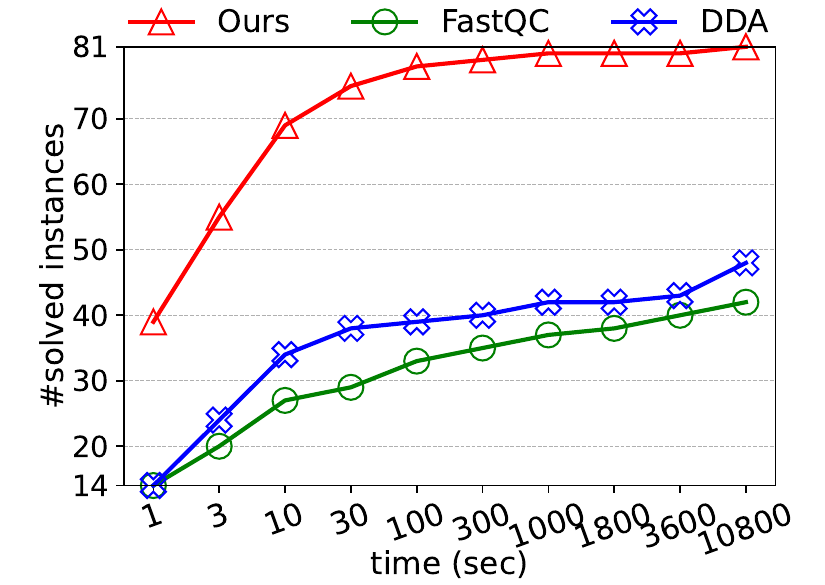}
        \label{fig:dimacs065}
    }
    \subfigure[$\gamma = 0.75$]{
        \includegraphics[width=0.21\textwidth]{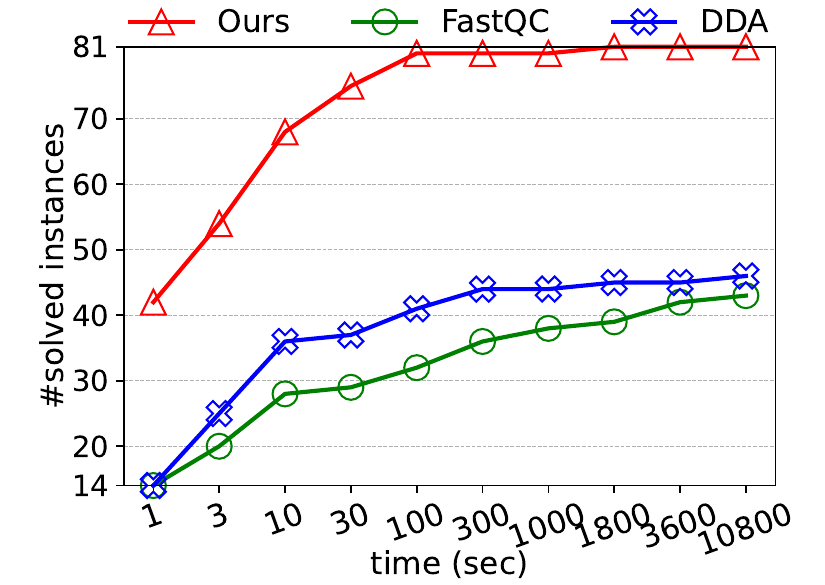}
        \label{fig:dimacs075}
    }
    \subfigure[$\gamma = 0.85$]{
        \includegraphics[width=0.21\textwidth]{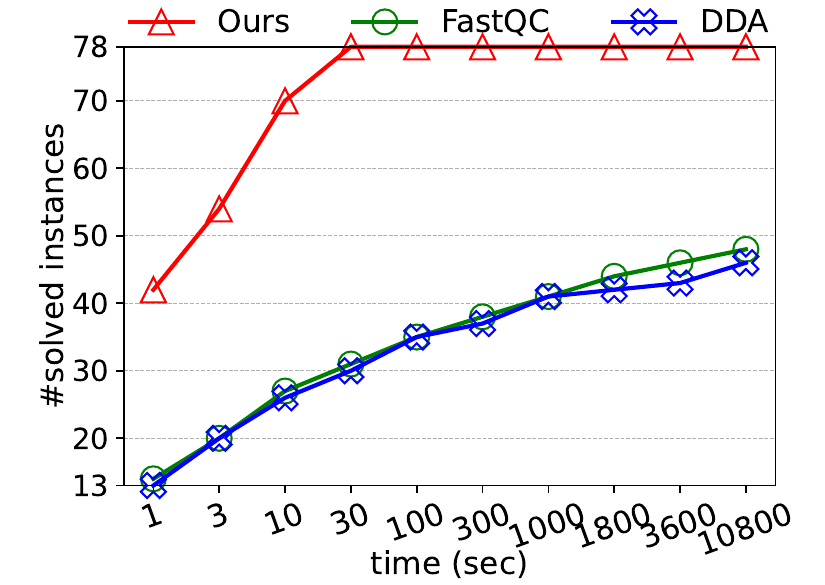}
        \label{fig:dimacs085}
    }
    \subfigure[$\gamma = 0.95$]{
        \includegraphics[width=0.21\textwidth]{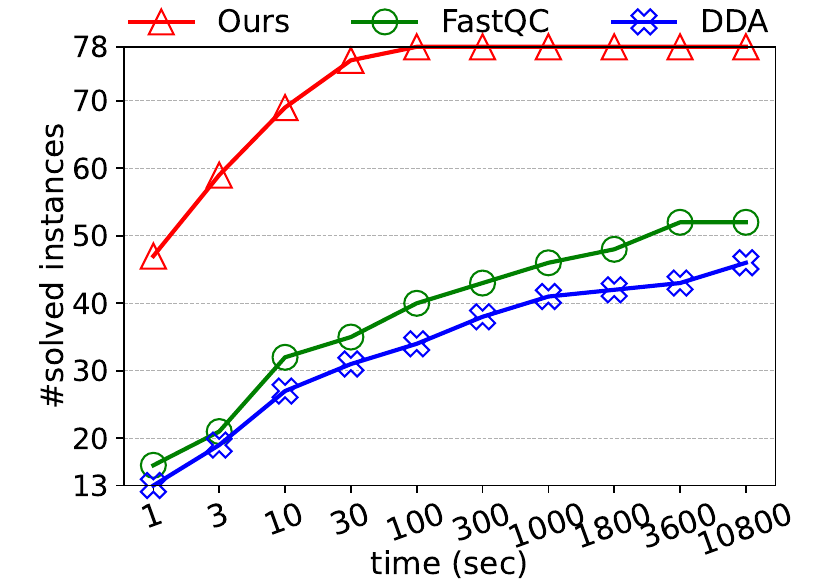}
        \label{fig:dimacs095}
    }
    \vspace{-0.2in}
    \caption{Number of solved instances on 10th DIMACS.}
    \label{Fig:dimacs10_various_gamma}
\end{figure}

\begin{figure}[t]
    \centering
    \subfigure[$\gamma = 0.65$]{
        \includegraphics[width=0.21\textwidth]{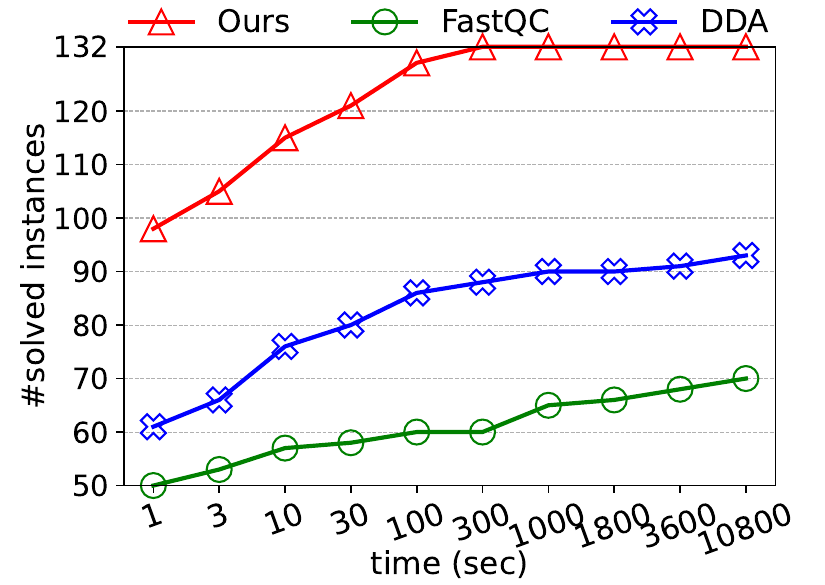}
        \label{fig:realworld065}
    }
    \subfigure[$\gamma = 0.75$]{
        \includegraphics[width=0.21\textwidth]{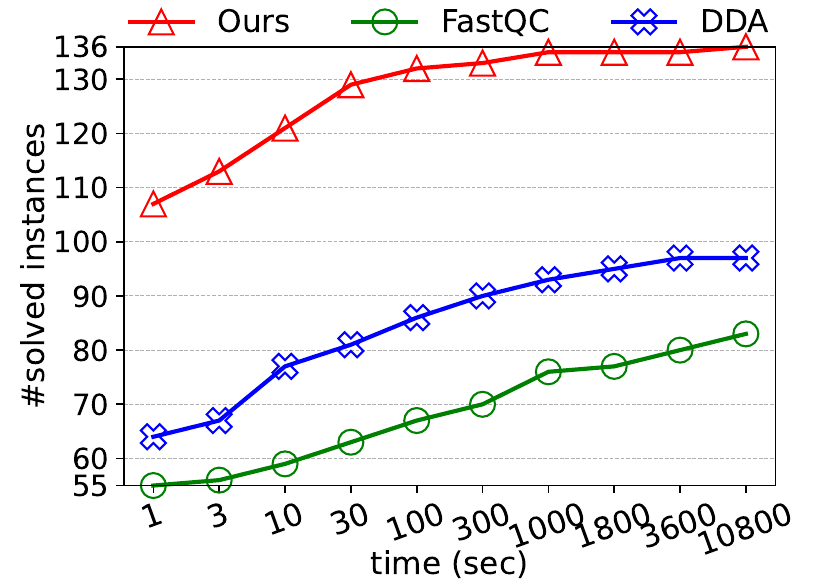}
        \label{fig:realworld075}
    }
    \subfigure[$\gamma = 0.85$]{
        \includegraphics[width=0.21\textwidth]{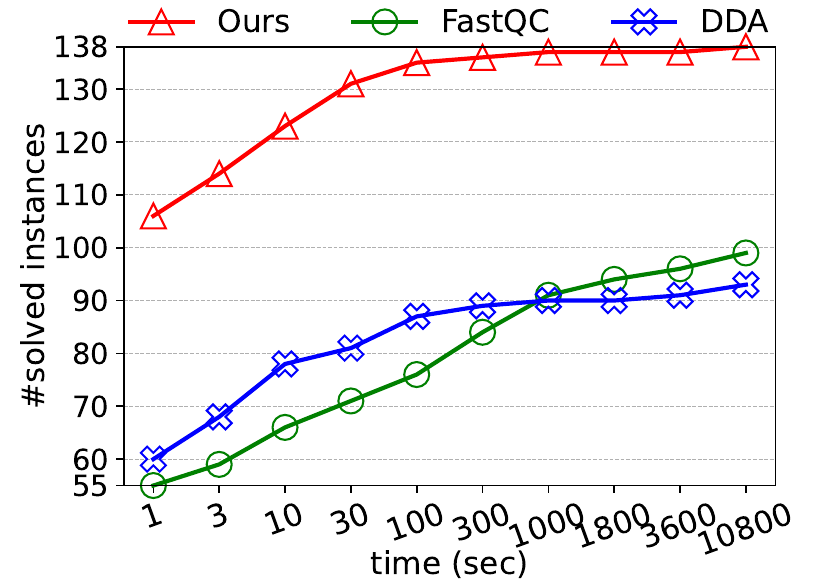}
        \label{fig:realworld085}
    }
    \subfigure[$\gamma = 0.95$]{
        \includegraphics[width=0.21\textwidth]{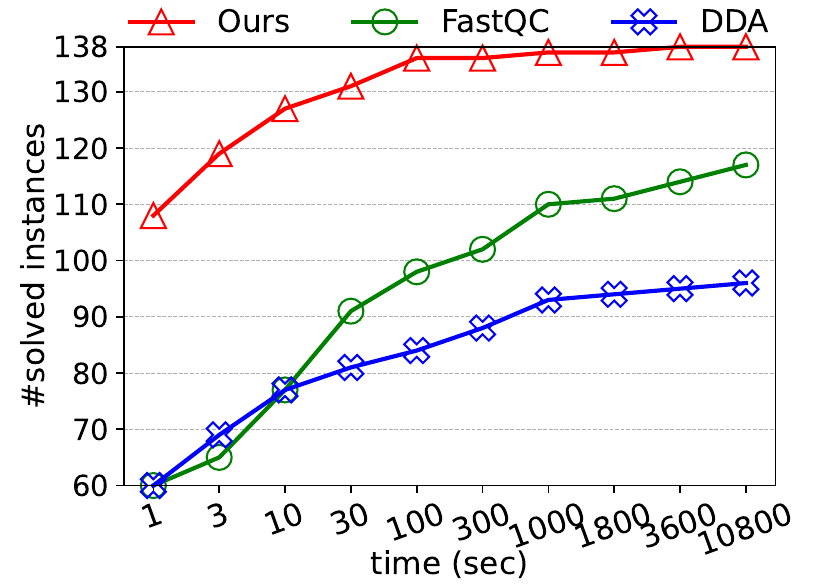}
        \label{fig:realworld095}
    }
    \vspace{-0.2in}
    \caption{Number of solved instances on real-world.}
    \label{Fig:realworld_various_gamma}
\end{figure}

\noindent \underline{\textbf{Number of solved instances.}} We first compare our algorithm \texttt{IterQC} with the baselines \texttt{DDA} and \texttt{FastQC} by considering the number of instances solved within 3-hour and 3-second limit on two collections in Figure~\ref{Fig:3h}.  In addition, we present the number of instances solved over time for both collections at $\gamma$ values of 0.65, 0.75, 0.85, and 0.95, in Figures~\ref{Fig:dimacs10_various_gamma} and~\ref{Fig:realworld_various_gamma}.
We have the following observations.  
\textbf{First}, our algorithm \texttt{IterQC} solves a greater number of instances across different values of $\gamma$, compared to the baselines \texttt{FastQC} and \texttt{DDA}. For example, in the real-world dataset with $\gamma = 0.7$ (in Figure~\ref{fig:realworld3h}), \texttt{IterQC} solves 135 out of 139 instances, while \texttt{DDA} and \texttt{FastQC} only solve 95 and 76 instances, respectively.
\textbf{Second}, in general, as the value of $\gamma$ decreases, \texttt{IterQC} tends to use relatively larger values of $k$ during each iteration, leading to higher computational costs and longer overall running times. As shown in Figures~\ref{fig:realworld3h},~\ref{fig:dimacs3s} and~\ref{fig:realworld3s}, the number of instances solved by \texttt{IterQC} generally increases with $\gamma$. However, this trend is less apparent in Figure~\ref{fig:dimacs3h}. Specifically, as $\gamma$ decreases from 0.8 to 0.7, \texttt{IterQC} solves more instances. This phenomenon may be due to the following factors. 
\texttt{IterQC} computes a heuristic upper bound $ub$ in the preprocessing stage. The gap between this $ub$ and the optimum solution $s^*$ is unpredictable for different values of $\gamma$. A smaller gap may result in a potentially shorter overall running time.
Additionally, in our graph reduction process in the preprocessing stage, a smaller $\gamma$ leads to a non-decreasing lower bound $lb$, which enhances the effectiveness of graph reduction and reduces subsequent search costs. Moreover, the increased computational complexity associated with smaller values of \(\gamma\) primarily arises from the branch-and-bound search process. Intuitively, as \(\gamma\) decreases, the relaxation of the clique condition becomes more significant, making it harder to prune branches that could previously be terminated early. The pseudo LB technique accelerates the branch-and-bound approach and helps reduce the increased difficulty introduced by smaller values of \(\gamma\).
\textbf{Third}, we observe in Figures~\ref{Fig:dimacs10_various_gamma} and~\ref{Fig:realworld_various_gamma} that the number of instances that \texttt{IterQC} can solve within 3 seconds exceeds the number solved by the other two baselines within three hours.
For example, on the 10th DIMACS graphs with $\gamma = 0.75$, \texttt{IterQC} solves 54 instances within 3 seconds, while \texttt{DDA} and \texttt{FastQC} solve 46 and 43 instances within 3 hours, respectively. Moreover, as shown in Figure~\ref{fig:realworld085}, \texttt{IterQC} can complete 105 instances in 1 second, while \texttt{DDA} and \texttt{FastQC} finish 93 and 99 instances within 3 hours, respectively.

\noindent \underline {\textbf{Performance on representative instances.}} 
The runtime performance comparison between \texttt{IterQC} and the two baseline algorithms with $\gamma = 0.75$ across 30 representative instances is shown in Table~\ref{table:representitive_75}. 
As illustrated in the table, \texttt{IterQC} consistently demonstrates superior efficiency, outperforming both baselines \texttt{FastQC} and \texttt{DDA} across nearly all instances. In particular, \texttt{IterQC} can solve all the graph instances, while both baseline algorithms \texttt{FastQC} and \texttt{DDA} exhibit a high occurrence of timeouts, failing to yield solutions within the 3-hour limit. Specifically, \texttt{FastQC} and \texttt{DDA} fail to solve 23 and 16 instances, respectively. Furthermore, \texttt{IterQC} successfully solves 14 out of the 30 representative instances where both baseline algorithms exceed the 3-hour time limit. For example, on G3, \texttt{IterQC} uses only 0.37 second, while both baselines cannot complete in 3 hours, achieving at least a 29,000$\times$ speed-up. These results further suggest the efficiency superiority of \texttt{IterQC} over both baseline algorithms. The superior performance of \texttt{IterQC}, particularly compared to \texttt{FastQC}, may be due to the hereditary property of the \( k \)-plex, which allows for more efficient pruning during the branch-and-bound search.
However, in rare cases, the computational overhead of \texttt{IterQC} exceeds that of \texttt{DDA}, such as in G19, G23, and G26. This is because \texttt{DDA} adopts an iterative approach based on the IP solver CPLEX, which differs fundamentally from the branch-and-bound-based approaches used by \texttt{IterQC} and \texttt{FastQC}. As a general-purpose solver, CPLEX is not specifically optimized for the $\gamma$-quasi-clique problem and is difficult to tailor for it. Consequently, while \texttt{DDA} may occasionally perform better in specific instances, \texttt{IterQC} generally outperforms \texttt{DDA} in most cases.

\begin{table}[t]
    \centering    
    \begin{minipage}{0.48\textwidth}
        \centering
        \captionsetup{font=small}
        \caption{Runtime performance (in seconds) of \texttt{IterQC}, \texttt{FastQC}, and \texttt{DDA} on 30 instances with $\gamma = 0.75$.}
        \vspace{-0.15in}
        \label{table:representitive_75}
        \scalebox{0.8}{
        \begin{tabular}{l|l l l|l|l l l}
        \hline
        \textbf{ID} & \texttt{IterQC} & \texttt{FastQC} & \texttt{DDA} & \textbf{ID} & \texttt{IterQC} & \texttt{FastQC} & \texttt{DDA} \\ \hline
        G1 & \textbf{0.03} & OOT & 0.16 & G16 & \textbf{5.70} & OOT & OOT \\ 
        G2 & \textbf{0.004} & 652.72 & 0.10 & G17 & \textbf{0.24} & 1073.59 & 0.34 \\ 
        G3 & \textbf{0.37} & OOT & OOT & G18 & \textbf{341.59} & OOT & OOT \\ 
        G4 & \textbf{0.02} & 4.93 & 69.71 & G19 & 4.06 & OOT & \textbf{2.73} \\ 
        G5 & \textbf{0.91} & OOT & OOT & G20 & \textbf{1.38} & OOT & 1.49 \\ 
        G6 & \textbf{4.46} & 80.61 & OOT & G21 & \textbf{1.87} & OOT & 2.14 \\ 
        G7 & \textbf{0.24} & OOT & 3059.20 & G22 & \textbf{8.70} & OOT & OOT \\ 
        G8 & \textbf{0.32} & OOT & OOT & G23 & 3.98 & OOT & \textbf{3.45} \\ 
        G9 & \textbf{0.02} & 3.88 & 0.20 & G24 & \textbf{13.07} & OOT & OOT \\ 
        G10 & \textbf{2.46} & OOT & OOT & G25 & 18.00 & OOT & \textbf{9.30} \\ 
        G11 & \textbf{17.71} & OOT & OOT & G26 & \textbf{14.86} & OOT & OOT \\ 
        G12 & \textbf{0.41} & OOT & OOT & G27 & \textbf{18.64} & OOT & OOT \\ 
        G13 & \textbf{0.26} & 34.60 & 297.66 & G28 & \textbf{4.75} & OOT & 42.54 \\ 
        G14 & \textbf{3.28} & 106.46 & OOT & G29 & \textbf{9.377} & OOT & 9.378 \\ 
        G15 & \textbf{772.93} & OOT & OOT & G30 & \textbf{25.50} & OOT & OOT \\ \hline
        \end{tabular}}
    \end{minipage}%
\end{table}

\noindent \underline {\textbf{Scalability test.}} We use G30 for the scalability test, which has the most vertices among the representative graphs. In our experiment, we randomly extract 20\% to 100\% of the vertices and test the performance of \texttt{IterQC} and two baselines with $\gamma$ values of 0.65, 0.75, 0.85, and 0.95 in Figure~\ref{Fig:scale_various_gamma}. The results demonstrate two findings: First, in almost all cases, \texttt{IterQC} consistently achieves the shortest runtime.
Second, across all four different values of $\gamma$, as the ratio increases—indicating a larger graph size -- the increase in the runtime of \texttt{IterQC}  is significantly smaller compared to the other algorithms. For example, in Figure~\ref{fig:scale085}, when the ratio increases from 0.2 to 0.8, the runtime of \texttt{DDA} rises from 0.52 seconds to 1817.60 seconds, whereas \texttt{IterQC} only increases from 0.49 seconds to 8.36 seconds. As for \texttt{FastQC}, at a ratio of 0.2, the runtime is 1086.88 seconds, but when the ratio reaches 0.4, it exceeds the time limit (10,800 seconds). These results demonstrate the scalability of \texttt{IterQC}.

\begin{figure}[t]
    \centering
    \subfigure[$\gamma = 0.65$]{
        \includegraphics[width=0.21\textwidth]{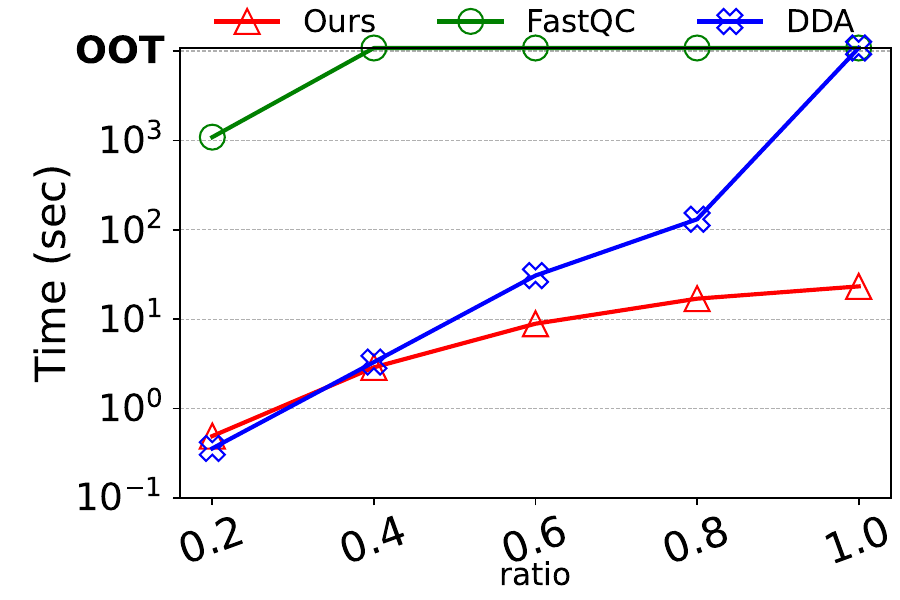}
        \label{fig:scale065}
    }
    \subfigure[$\gamma = 0.75$]{
        \includegraphics[width=0.21\textwidth]{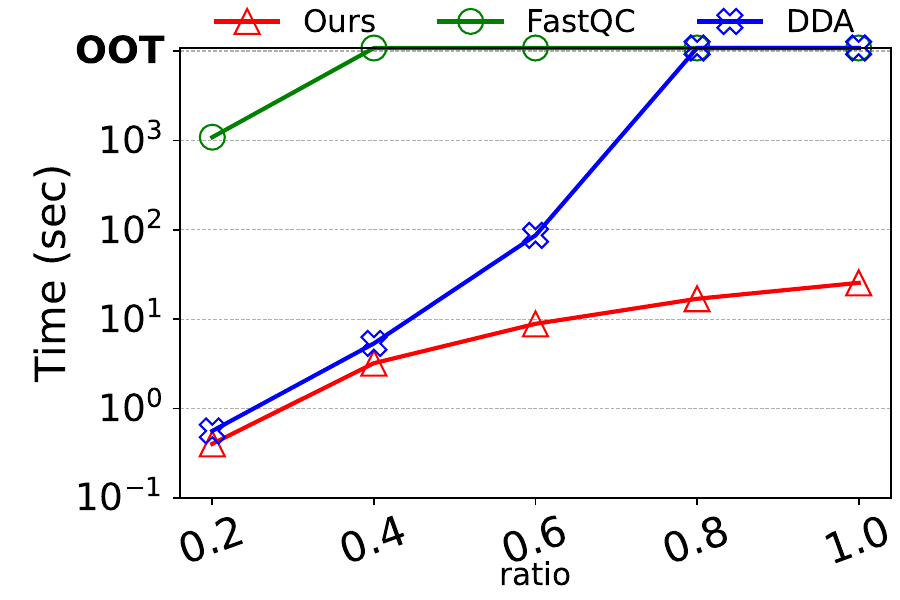}
        \label{fig:scale075}
    }
    \subfigure[$\gamma = 0.85$]{
        \includegraphics[width=0.21\textwidth]{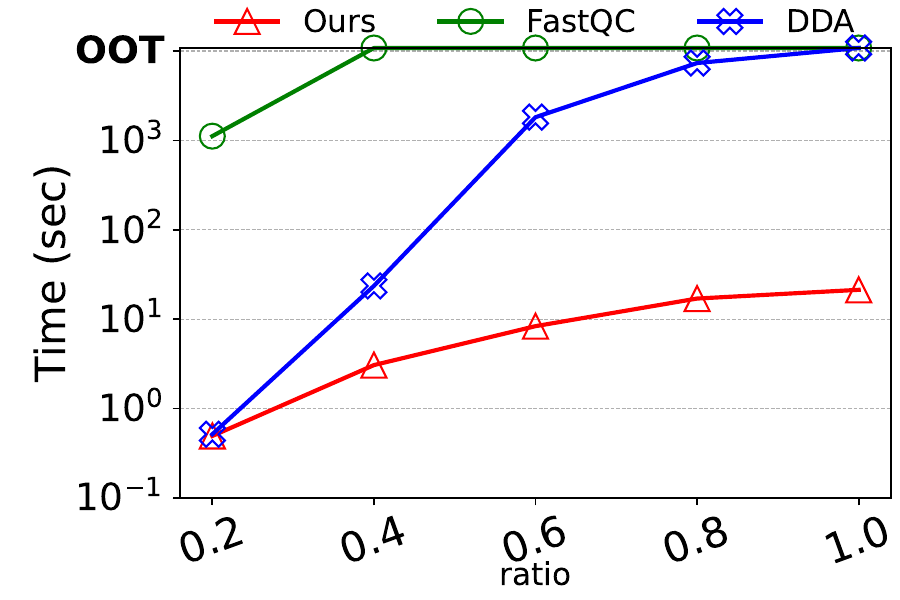}
        \label{fig:scale085}
    }
    \subfigure[$\gamma = 0.95$]{
        \includegraphics[width=0.21\textwidth]{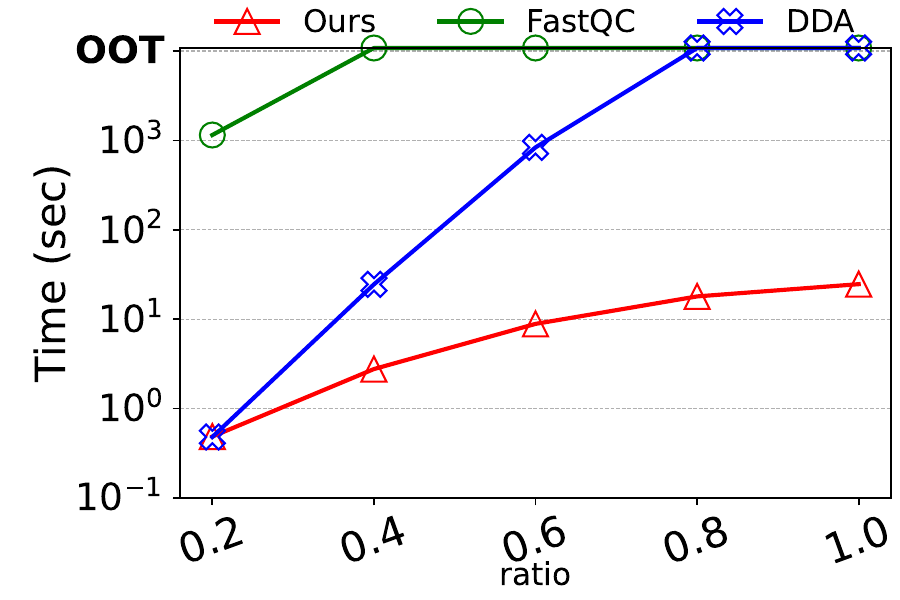}
        \label{fig:scale095}
    }
    \vspace{-0.2in}
    \caption{Scalability test on G30.}
    \label{Fig:scale_various_gamma}
\end{figure}

\subsection{Ablation Studies}
We conduct ablation studies to evaluate the effectiveness of the techniques of preprocessing and pseudo LB proposed in Section~\ref{sec:improved}.
We compare \texttt{IterQC} with the following variants: 
\begin{itemize}[leftmargin=*]
    \item {\texttt{IterQC-PP}:} it removes the preprocessing technique in \texttt{IterQC}, which includes (1) the initial estimation of lower and upper bounds, and (2) graph reduction. Specifically, \texttt{IterQC-PP} replaces Lines 1-3 in Algorithm~\ref{alg:improved-framwork} with $ub \gets |V|$.

    \item {\texttt{IterQC-PLB}:} it removes the pseudo lower bound $pseudo\mbox{-}lb$ and utilizes the true heuristic lower bound by replacing Line 3 in Algorithm~\ref{alg:kplex-search} with $pseudo\mbox{-}lb \gets lb\mbox{-}plex$.
\end{itemize} 

Table~\ref{table:Compare_Ablation} presents the runtime performance of \texttt{IterQC}, \texttt{IterQC-PP}, and \texttt{IterQC-PLB} on 30 representative instances with $\gamma = 0.75$. 

\begin{table}[t]
    \centering
    \captionsetup{font=small}
    \caption{Runtime performance (in seconds) of \texttt{IterQC}, \texttt{IterQC-PP}, and \texttt{IterQC-PLB} on 30 instances with $\gamma = 0.75$.}
    \vspace{-0.15in}
    \label{table:Compare_Ablation}
    \scalebox{0.8}{
    \begin{tabular}{c|c c c|c|c c c}
        \hline
        \textbf{ID} & \texttt{IterQC} & \texttt{-PP} & \texttt{-PLB} & \textbf{ID} & \texttt{IterQC} & \texttt{-PP} & \texttt{-PLB} \\ \hline
        G1 & \textbf{0.034} & 0.10 & 0.035 & G16 & \textbf{5.70} & 35.54 & 6.31 \\ 
        G2 & \textbf{0.0043} & 0.039 & 0.0044 & G17 & \textbf{0.24} & 2.96 & 0.32 \\ 
        G3 & \textbf{0.37} & 2.21 & 0.41 & G18 & \textbf{341.59} & 4946.34 & OOT \\ 
        G4 & 0.019 & 0.08 & \textbf{0.015} & G19 & \textbf{4.06} & 16.96 & 4.14 \\ 
        G5 & 0.91 & 6.64 & \textbf{0.80} & G20 & \textbf{1.38} & 17.11 & 1.73 \\ 
        G6 & \textbf{4.46} & 6.65 & 4.86 & G21 & \textbf{1.87} & 29.47 & 2.50 \\ 
        G7 & 0.24 & 1.36 & \textbf{0.17} & G22 & \textbf{8.70} & 49.07 & 9.29 \\ 
        G8 & \textbf{0.32} & 1.24 & 0.61 & G23 & \textbf{3.98} & 37.37 & 5.15 \\ 
        G9 & 0.021 & 0.17 & \textbf{0.018} & G24 & \textbf{13.07} & 53.49 & 14.76 \\ 
        G10 & \textbf{2.46} & 2.51 & 29.79 & G25 & \textbf{18.00} & 81.90 & 18.60 \\ 
        G11 & \textbf{17.71} & 97.73 & 134.01 & G26 & \textbf{14.86} & 113.84 & 16.36 \\ 
        G12 & \textbf{0.41} & 1.05 & 0.66 & G27 & \textbf{18.64} & 79.77 & 21.07 \\ 
        G13 & 0.26 & 0.99 & \textbf{0.22} & G28 & \textbf{4.75} & 87.08 & 4.94 \\ 
        G14 & \textbf{3.28} & 11.15 & 3.42 & G29 & \textbf{9.38} & 191.62 & 9.99 \\ 
        G15 & \textbf{772.93} & 1385.18 & OOT & G30 & \textbf{25.50} & 195.29 & 25.97 \\ \hline
    \end{tabular}
    }
\end{table}

\noindent \underline{\textbf{Effectiveness of the preprocessing technique.}} From Table~\ref{table:Compare_Ablation}, we observe that \texttt{IterQC} consistently outperforms \texttt{IterQC-PP}, achieving a speedup factor of at least 5 in 17 instances and at least 10 in 6 instances, with a remarkable speedup factor of 20.32 on G29.
We also summarize additional information for our preprocessing technique in Table~\ref{table:Preproccess}, which details the percentages of vertices and edges pruned during preprocessing, as well as the lower and upper bounds ($lb$ and $ub$) for the optimum solution \( s^* \). 
From Table~\ref{table:Preproccess}, we observe that in 3 instances (G17, G20, and G29), the preprocessing technique prunes all vertices and edges, effectively obtaining the solution directly, while in 14 instances, it removes at least 90\% of the vertices.
Moreover, across the 30 instances, the preprocessing technique enables the iteration process to start from a smaller initial value (closer to the optimum solution $s^*$), as indicated by the upper bound $ub$ (in contrast to the trivial upper bound $|V|$ in Table~\ref{table:representive_instances}).

\begin{table}[t]
    \centering
    \captionsetup{font=small}
    \caption{Preprocessing information with $\gamma = 0.75$, where Red-V/Red-E represent the percentages of reduced vertices/edges.}
    \vspace{-0.15in}
    \label{table:Preproccess}
    \scalebox{0.8}{
    \begin{tabular}{c|c|c|c|c|c}
    \hline
        \textbf{ID} & Red-V (\%) & Red-E (\%) & $lb$ & $ub$ & $s^*$\\ 
        \hline
        G1 & 98.27 & 87.46 & 116 & 119 & 117 \\ 
        G2 & 79.85 & 59.63 & 230 & 233 & 230 \\ 
        G3 & 64.76 & 44.58 & 74 & 87 & 74 \\ 
        G4 & 98.13 & 82.39 & 17 & 31 & 21 \\ 
        G5 & 58.88 & 37.46 & 57 & 84 & 66 \\ 
        G6 & 1.07 & 0.57 & 24 & 48 & 30 \\ 
        G7 & 92.72 & 75.63 & 54 & 71 & 59 \\ 
        G8 & 96.73 & 67.96 & 57 & 91 & 58 \\ 
        G9 & 83.63 & 79.28 & 6 & 7 & 6 \\ 
        G10 & 95.91 & 62.17 & 29 & 79 & 33 \\ \hline
        G11 & 39.06 & 7.39 & 13 & 124 & 45 \\ 
        G12 & 97.86 & 85.13 & 37 & 69 & 45 \\ 
        G13 & 0.00 & 0.00 & 3 & 7 & 6 \\ 
        G14 & 0.00 & 0.00 & 6 & 13 & 10 \\ 
        G15 & 97.46 & 53.76 & 86 & 316 & 127 \\ 
        G16 & 99.79 & 95.67 & 2209 & 2945 & 2209 \\ 
        G17 & 100.00 & 100.00 & 5 & 5 & 5 \\ 
        G18 & 75.40 & 50.63 & 66 & 308 & 130 \\ 
        G19 & 99.53 & 99.45 & 20 & 27 & 25 \\ 
        G20 & 100.00 & 100.00 & 5 & 5 & 5 \\ \hline
        G21 & 0.00 & 0.00 & 3 & 5 & 5 \\ 
        G22 & 0.00 & 0.00 & 2 & 4 & 2 \\ 
        G23 & 0.00 & 0.00 & 2 & 5 & 5 \\ 
        G24 & 0.00 & 0.00 & 2 & 4 & 2 \\ 
        G25 & 99.99 & 99.99 & 25 & 28 & 27 \\ 
        G26 & 0.00 & 0.00 & 5 & 7 & 6 \\ 
        G27 & 0.00 & 0.00 & 2 & 4 & 2 \\ 
        G28 & 0.00 & 0.00 & 3 & 5 & 5 \\ 
        G29 & 100.00 & 100.00 & 5 & 5 & 5 \\ 
        G30 & 99.15 & 96.69 & 10 & 23 & 10 \\ \hline
    \end{tabular}
    }
\end{table}

\noindent \underline{\textbf{Effectiveness of the pseudo LB technique.}} We can see in Table~\ref{table:Compare_Ablation} that applying the pseudo LB technique leads to an improved performance in 25 of these 30 instances. Moreover, compared to \texttt{IterQC-PLB}, \texttt{IterQC} successfully solves 2 additional \textbf{OOT} instances, i.e., G15 and G18. For G18, \texttt{IterQC} solves in 341.59 seconds while \texttt{IterQC-PLB} times out, implying a speedup factor of at least 31.62. 
This improvement is due to the acceleration of the branch-and-bound search by the pseudo LB technique in Algorithm~\ref{alg:kplex-search}, particularly by leveraging the graph structure: dense local regions increase branch-and-bound search costs, leading to greater speedups.  
Conversely, instances like G4, G5, G7, G9, and G13 show lower effectiveness when the running time is dominated by the computations of the heuristic lower bound in Line 4 of Algorithm~\ref{alg:kplex-search}. Despite this, even in these instances, the impact on running time is minor, with all such instances completing in under 1 second.  

\section{Related Work}\label{sec:related-work}
\noindent \underline{\textbf{Maximum $\gamma$-quasi-clique search problem.}}
The maximum $\gamma$-quasi-clique search problem is NP-hard~\cite{MIH99,Pastukhov2018gamma} and W[1]-hard parameterized by several graph parameters~\cite{Baril2021hardness,Baril2024hardness}.
The state-of-the-art exact algorithm for solving this problem is \texttt{DDA}~\cite{Pastukhov2018gamma} and extensively discussed in Section~\ref{sec-intro}.
In contrast, Bhattacharyya and Bandyopadhyay~\cite{BB09} provided a greedy heuristic. Further studies~\cite{Chou2015oneVerticeQC,LL16} addressed related problems of finding the largest maximal quasi-cliques that include a given target vertex or vertex set in a graph. Moreover, Marinelli et al.~\cite{marinelli2021lp} proposed an IP-based method to compute upper bounds for the maximum \( \gamma \)-quasi-clique.

\noindent \underline{\textbf{Maximal $\gamma$-quasi-clique enumeration problem.}}
A closely related problem is the enumeration of all maximal $\gamma$-quasi-cliques in a given graph~\cite{LW08,khalil2022parallel,yu2023fast}, where a $\gamma$-QC $g$ is \emph{maximal} if no supergraph $g'$ of $g$ is also a $\gamma$-QC. Several branch-and-bound algorithms have been proposed to tackle this problem by using multiple pruning techniques to reduce the search space during enumeration. In particular, Liu and Wong~\cite{LW08}, Guo et al.~\cite{guo2020scalable}, and Khalil et al.~\cite{khalil2022parallel} developed such algorithms to improve efficiency. Recently, Yu and Long~\cite{yu2023fast} introduced \texttt{FastQC}, the current state-of-the-art algorithm combining pruning and branching co-design approach. 
We remark that the maximum $\gamma$-QC search problem can be solved using algorithms designed for maximal $\gamma$-QC enumeration, as the maximum $\gamma$-QC is always a maximal one in the graph. In our experimental studies, we adapt the state-of-the-art maximal $\gamma$-QC enumeration algorithm, \texttt{FastQC}, as the baseline method to solve our problem.
Some studies explored different problem variants. For example, Sanei-Mehri et al.~\cite{sanei-mehri2021largest} focused on the top-$k$ variant, Guo et al.~\cite{guo2022directed} studied the problem in directed graphs, and others considered graph databases instead of a single graph~\cite{JP09,zeng2007out}. 

\noindent \underline{\textbf{Other cohesive subgraph mining problems.}}  
Another approach to cohesive subgraph mining involves relaxing the clique definition from the perspective of edges. This approach gives rise to the concept of the edge-based $\gamma$-quasi-clique~\cite{Abello2002egde,CondeCespedes2018edge,Pattillo2013edge}, which is also referred to as pseudo-cliques~\cite{Uno2010}, dense subgraphs~\cite{Long2010}, or near-cliques~\cite{Brakerski2011, Tadaka2016}. In this cohesive subgraph model, the total number of edges in a subgraph must be at least $\gamma \cdot \binom{n}{2}$. Very recently, Rahman et al.~\cite{rahman2024pseudo} introduced a novel pruning strategy based on Tur\'{a}n's theorem~\cite{Jain2017} to obtain an exact solution, building on the PCE algorithm proposed by Uno~\cite{Uno2010}.
Similar to the $\gamma$-quasi-clique problem, several studies have also explored non-exact approaches to solve the edge-based $\gamma$-quasi-clique problem~\cite{Abello2002egde, Brunato2008, Chen2021,LiuZWZL24}. For instance, Tsourakakis et al.~\cite{Tsourakakis2013} proposed an objective function that unifies the concepts of average-degree-based quasi-cliques and edge-based $\gamma$-quasi-clique. 
There also exist many other types of cohesive subgraphs, including $k$-plex \cite{Dai2022kplex, Wang2022kplex,Zhou2020kplex,chang2022efficient, Gao2018, jiang2023refined, jiang2021, wang2023kplex, Xiao2017maxkplex, zhou2021improving}, $k$-defective clique \cite{Chang2023, Dai2023defect, Gao2022defect}, and densest subgraph \cite{Ma2021densest, Xu2024densest}.
Moreover, the topic of cohesive subgraphs has also been widely studied in other types of graphs, including bipartite graphs \cite{Chen2021bipartite, Dai2023bipartite, Luo2022bipartite, Yu2022bipartite, Yu2023bipartite,Yu2021biplex}, directed graphs \cite{Gao2024directed}, temporal graphs \cite{Bentert2019temp}, and uncertain graphs \cite{Dai2022uncertain}. For an overview on cohesive subgraphs, see the excellent books and survey, e.g.,~\cite{Chang2018cosub, Fang2020cosub, Fang2021cosub, Huang2019cosub, Lee2010cosub}.

\section{Conclusion}\label{sec:conclusion}
In this paper, we studied the maximum $\gamma$-quasi clique problem and proposed an iterative framework incorporating two novel techniques: the pseudo lower bound and preprocessing. Extensive experiments demonstrated the superiority of our algorithm \texttt{IterQC} over state-of-the-art methods \texttt{DDA} and \texttt{FastQC}. In future work, we aim to extend our iterative approach to other cohesive graph models.

\begin{acks}
This research is partially supported by the National Natural Science Foundation of China (No. 62102117), by the Shenzhen Science and Technology Program (No. GXWD20231129111306002), and by the Key Laboratory of Interdisciplinary Research of Computation and Economics (Shanghai University of Finance and Economics), Ministry of Education.
This research is also partially supported by the National Natural Science Foundation of China (No. 62472125), the Natural Science Foundation of Guangdong Province, China (No. 2025A1515011258), and Shenzhen Sustained Support for Colleges \& Universities Program (No. GXWD20231128102922001). 
This research is also supported by the Ministry of Education, Singapore, under its Academic Research Fund (Tier 2 Award MOE-T2EP20221-0013 and Tier 1 Award (RG20/24)). Any opinions, findings and conclusions or recommendations expressed in this material are those of the author(s) and do not reflect the views of the Ministry of Education, Singapore.
\end{acks}

\clearpage
\bibliographystyle{ACM-Reference-Format}
\bibliography{ref}

\appendix

\section{Omitted Proofs}
\subsection{Omitted Proofs in Section~\ref{sec:basic-iterative}}\label{sec:proof-basic}

\begin{proof}[Proof of Lemma~\ref{lem:input-qc}]
     When $G$ is a $\gamma$-quasi-clique, we have $d_{G}(v) \geq \gamma \cdot (|V| - 1)$, $\forall v \in V$. Since $d_G(v)$ is an integer for each $v \in V$, it follows that $d_{G}(v) \geq \lceil \gamma \cdot (|V| - 1) \rceil \geq |V| - (\lfloor (1 - \gamma) \cdot (|V| - 1)\rfloor + 1) = |V| - k_1$. Thus, $G$ is a $k_1$-plex and $\texttt{solve-k}(k_1) = |V|$, which implies that $s_1 = s_0$. At this point, we have $k_1 = \texttt{get-k}(s_1)$, and the algorithm results in $s_1 = |V|$, which completes the proof. 
\end{proof}

\begin{proof}[Proof of Lemma~\ref{lem:seq-s-non-identical}]
     Assume, to the contrary, that $s_i = s_{i+1}$ for $0 \leq i \leq p-1$. 
     Then, it follows that $k_{i+1} = \texttt{get-k}(s_i) = \texttt{get-k}(s_{i+1})$, which satisfies the termination condition in Line 4. This implies that the algorithm terminates at the \((i+1)\)-st iteration, thus \(i = p\), leading to a contradiction to $i \leq p - 1$.
\end{proof}

\begin{proof}[Proof of Corollary~\ref{cor:sequence}]
    By Lemma~\ref{lem:seq-s}, the sequence is strictly decreasing. Further, it is easy to see that $s_p \geq 1$ since an induced subgraph with a single vertex is a trivial solution to the maximum $k$-plex problem with any $k \geq 1$. 
    Moreover, as $\{s_0,s_1,\ldots,s_p\}$ is an integer sequence, the difference between any two consecutive elements is at least 1, which implies that $p\leq s_0 = |V|=n$.
\end{proof}

\begin{proof}[Proof of Lemma~\ref{lem:seq-s-least}]
    Recall that $s^*$ is the size of the optimum solution.
    Let $u$ be the vertex with the minimum degree in the largest $\gamma$-quasi-clique $g^*$. Let $\gamma^* = \frac{d_{g^*}(u)}{s^* - 1}$. Then we have 
    $s^* = \texttt{solve-plex}\left(1 + \left\lfloor (1 - \gamma^*) \cdot (s^* - 1) \right\rfloor \right)$. Additionally, from the definition of a $\gamma$-quasi-clique, it also follows that $\gamma^* \geq \gamma$.
    We then use the mathematical induction to prove.\\
    \textcircled{\scriptsize{1}} \normalsize\enspace $i=0$. $s_{i} = s_0 = |V| \geq s^* $. The base case holds true.\\
    \textcircled{\scriptsize{2}} \normalsize\enspace $i \geq 1$. Assume that the induction holds for $i-1$, i.e., $s_{i-1} \geq s^*$, it follows that 
    \begin{align*}
        s_{i} &= \texttt{solve-plex}(k_{i}) = \texttt{solve-plex}(\texttt{get-k}(s_{i-1})) \\
            &= \texttt{solve-plex}\left(1 + \left\lfloor (1-\gamma) \cdot (s_{i-1} - 1) \right\rfloor\right) \\
        &\geq \texttt{solve-plex}\left(1 + \left\lfloor (1 - \gamma^*) \cdot (s_{i-1} - 1) \right\rfloor\right) \\
        &\geq \texttt{solve-plex}\left(1 + \left\lfloor (1-\gamma^*) \cdot (s^* - 1) \right\rfloor\right) = s^*.
    \end{align*}
Based on \textcircled{\scriptsize{1}}, \textcircled{\scriptsize{2}}, and the principle of mathematical induction, we complete the proof of Lemma~\ref{lem:seq-s-least}.
\end{proof}

\subsection{Correctness Proof of Algorithm~\ref{alg:improved-iterative-search}}\label{sec:proof-improved}

We first show a property on the sequence generated by our basic iterative framework in Algorithm~\ref{alg:basic-framework}. To simplify the proofs, in the following discussion, let \( s_p \) denote the last element of the sequence \( \{s_0, s_1, \ldots, s_p\} \) generated by the iterative frameworks (either Algorithm~\ref{alg:basic-framework} or Algorithm~\ref{alg:improved-iterative-search}). It is important to highlight that, in contrast to \(s_p\) in the correctness proof of the basic iterative framework, which represents the penultimate computed value, here \( s_p \) specifically corresponds to the value at the iteration where Algorithm~\ref{alg:kplex-search} terminates, i.e., \( k_p = \texttt{get-k}(s_p)\).

\begin{lemma}\label{cor:iteration with ub}
    For any integer $\overline{s} \geq s^*$, the last element $s_p$ of the sequence $\{\overline{s},s_1,s_2,\ldots,s_p\}$ generated by Algorithm~\ref{alg:basic-framework} is equal to $s^*$.
\end{lemma}
    \begin{proof}
        The proof is similar to the proofs of Lemmas~\ref{lem:seq-s-least} and~\ref{lem:correct-ans}, which are based on a trivial upper bound, i.e., $|V| \geq s^*$. By using $\overline{s}$ instead of $|V|$, we can prove Lemma~\ref{cor:iteration with ub}.
    \end{proof}

 Subsequently, in Lemma~\ref{lem:connection}, we make a connection between the sequences generated by Algorithm~\ref{alg:basic-framework} and Algorithm~\ref{alg:improved-iterative-search}. 
 
\begin{lemma}\label{lem:connection}
    Consider the $i$-th iteration of Lines 2-6 in Algorithm~\ref{alg:improved-iterative-search}, where \texttt{Plex-Search}$(G,k,ub\mbox{-}plex)$ is called in Line 3. We let $\overline{s} \gets \texttt{solve-plex}(k)$, i.e., $\overline{s}$ is the size of maximum $k$-plex. If we have $k=\texttt{get-k}(ub\mbox{-}plex)$ and $ub\mbox{-}plex = s_{i-1} \geq s^*$, then $s_i \geq \overline{s} \geq s^*$.
\end{lemma}
\begin{proof}
    In Algorithm~\ref{alg:kplex-search}, we invoke the branch-and-bound algorithm \texttt{Plex-BRB} with a pseudo lower bound of $pseudo\mbox{-}lb$ in Line 4, where $pseudo\mbox{-}lb = \lfloor (lb\mbox{-}plex + ub\mbox{-}plex)/2 \rfloor$ is not the true lower bound. In other words, with $pseudo\mbox{-}lb$, the size of the returned vertex set $S$ may not be larger than our pseudo lower bound $pseudo\mbox{-}lb$. We have two possible cases as follows.  
    \begin{enumerate}[leftmargin=*]
        \item $\overline{s}>pseudo\mbox{-}lb$. In this case, since the pruning in \texttt{Plex-BRB} using $pseudo\mbox{-}lb$ does not affect the generation of the correct solution of maximum $k$-plex, the returned vertex set $S$ from \texttt{Plex-BRB} corresponds to the maximum $k$-plex. Thus, we have $|S| = \overline{s}$. Moreover, as $s_i = \max\{pseudo\mbox{-}lb, |S|\}$, we know that $s_i \geq \overline{s}$.
        \item $\overline{s} \leq pseudo\mbox{-}lb$. In this case, in Line 4 of Algorithm~\ref{alg:improved-iterative-search}, we have $s_i = \max\{pseudo\mbox{-}lb, |S|\}$, which directly implies that $s_i \geq \overline{s}$.
    \end{enumerate}
    
    In either case mentioned above, we have $s_i \geq \overline{s}$. According to Lemmas \ref{lem:seq-s-least} and \ref{cor:iteration with ub} and $ub\mbox{-}plex = s_{i-1} \geq s^*$, we know that $\overline{s} \geq s^*$. The proof is thus complete.
\end{proof}

Utilizing the connection, we verify the correctness of Algorithm~\ref{alg:improved-iterative-search} that uses the technique of pseudo LB in the following.
\begin{lemma}\label{lem:improved-correctness}
     Algorithm~\ref{alg:improved-iterative-search} correctly finds the largest $\gamma$-quasi-clique.
\end{lemma}
\begin{proof}
    According to Lemma~\ref{lem:connection}, we can guarantee that $s_i \geq s^*$ for each iteration $i$ in Lines 2-6 of Algorithm~\ref{alg:improved-iterative-search}. Thus, we can view Algorithm~\ref{alg:improved-iterative-search} as a procedure that continuously seeks upper bounds for the maximum $\gamma$-quasi-clique across all iterations. Let \( s_p \) denote the last element of the sequence \( \{s_0, s_1, \ldots, s_p\} \) generated by Algorithm~\ref{alg:improved-iterative-search}. Next, we prove that $s_p$ is equal to $s^*$.

    Consider the $i$-th iteration in Lines 2-6 of Algorithm~\ref{alg:improved-iterative-search}, where \texttt{Plex-Search}$(G,k,ub\mbox{-}plex)$ is called in Line 3. Note that $ub\mbox{-}plex = s_{i-1}$. We focus on this \texttt{Plex-Search} and have two cases.  
    \begin{enumerate}[leftmargin=*]
        \item $lb\mbox{-}plex = ub\mbox{-}plex$. Since $ub\mbox{-}plex = s_{i-1}$, it follows that $ub\mbox{-}plex = s_{i-1} \geq s^*$. We also know that there exists a $k$-plex of size $lb\mbox{-}plex$ computed in Line 1 of Algorithm~\ref{alg:kplex-search}. Thus, we have a $k$-plex of size $s_{i-1}$, where $k = \texttt{get-k}(s_{i-1})$. Similar to Lemma~\ref{lem:correct-ans}, it is easy to conclude that the returned result in Line 5 of Algorithm~\ref{alg:improved-iterative-search} is exactly $s^*$. Thus, $s^* = s_{i-1} = s_i$.
        \item  $lb\mbox{-}plex < ub\mbox{-}plex$. It is easy to see that $pseudo\mbox{-}lb < ub\mbox{-}plex$. If $s_i = pseudo\mbox{-}lb$ in Line 3 of Algorithm~\ref{alg:improved-iterative-search}, it guarantees that $s_i < s_{i-1}$. Otherwise, i.e., $s_i = pseudo\mbox{-}size$, we use $\overline{s}_i$ to denote the size of the largest $k$-plex in the graph $G$,  according to the proof of Lemma~\ref{lem:connection}, $s_i = \overline{s_i}$, which still satisfies $s_i = \overline{s_i} < s_{i-1}$. Thus, in any case, we have $s_i < s_{i-1}$.
    \end{enumerate}
    Since $\{s_0,s_1,\ldots,s_p\}$ is an integer sequence with each $s_i \geq s^*$, the case in 2) is always finite. In other words, the case in 1) will definitely occur. The proof is complete.
\end{proof}

\section{Omitted Examples}\label{sec:example}
\noindent \underline{\textbf{Example of the basic iterative framework.}}
\begin{figure}[t]
    \centering
    \subfigure[An example graph.]{
        \includegraphics[height=0.16\textwidth,width=0.256\textwidth, trim=6cm 7.5cm 15.2cm 4cm, clip]{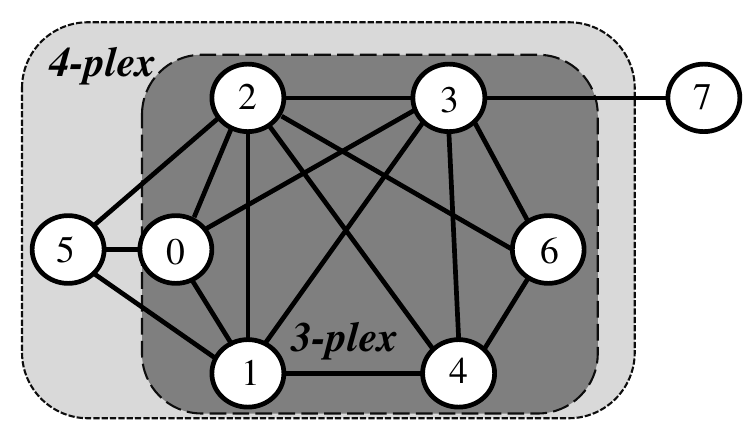} 
        \label{Fig:basic_framework_example}
    }
    \subfigure[Illustration of the example.]{
    \includegraphics[height=0.16\textwidth,width=0.16\textwidth, trim=6cm 22.5cm 9cm 1cm, clip]{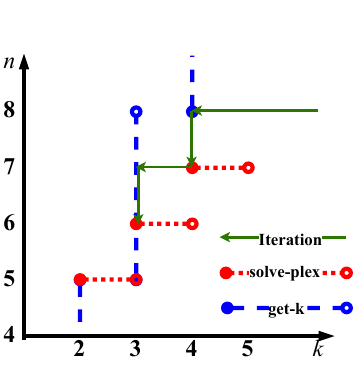} 
        \label{Fig:basic_framework_example2}
    }
    \vspace{-0.2in}
    \caption{An example of Algorithm~\ref{alg:basic-framework} with $\gamma = 0.55$.}\label{Fig:basic_framework}
\end{figure}
We present an example of Algorithm~\ref{alg:basic-framework} and illustrate its iterative process in Figures~\ref{Fig:basic_framework_example} and~\ref{Fig:basic_framework_example2}, respectively. 
In this example with $\gamma = 0.55$, the graph has 8 vertices, which means $s_0 = 8$. By applying $k_1 = \texttt{get-k}(s_0) = \lfloor (1 - 0.55) \cdot (8 - 1)\rfloor + 1 = 4$, the function \texttt{solve-plex} identifies the largest 4-plex, highlighted by the light-colored box in Figure~\ref{Fig:basic_framework_example}, which results in $s_1 = 7$. Subsequently, $k_2 = \texttt{get-k}(s_1) = \lfloor (1 - 0.55) \cdot (7 - 1)\rfloor + 1=3$, and solving \texttt{solve-plex} produces the largest 3-plex, indicated by the dark-colored box in Figure~\ref{Fig:basic_framework_example}, with $s_2 = 6$ vertices. At this point, the stopping condition $k_2 = \texttt{get-k}(s_2)$ is satisfied, and the iteration terminates. The maximum 3-plex shown in Figure~\ref{Fig:basic_framework_example} corresponds to the desired maximum $0.55$-quasi-clique.  
In Figure~\ref{Fig:basic_framework_example2}, the horizontal axis represents \(k\), while the vertical axis represents \(n\). Moreover, (1) the sparse light blue dashed line shows the relationship between \(n\) and the corresponding \(k\) obtained from \texttt{get\mbox{-}k}; (2) the dense red dashed line shows the relationship between \(k\) and the resulting \(n\) obtained from \texttt{solve\mbox{-}plex}; (3) the green solid line represents the overall iterative process of Algorithm~\ref{alg:basic-framework}, demonstrating how \(n\) evolves from the initial graph of size 8 to the final result of the maximum $\gamma$-quasi-clique with size 6. Specifically, we start with \(s_0 = 8\). After the first iteration, applying \texttt{get-k}(\(s_0\)) transitions to the state to \((4, 8)\) in Figure~\ref{Fig:basic_framework_example2} on the corresponding axes (i.e., \(k_1 = 4\), \(s_0 = 8\)). Then, solving \(s_1 = \texttt{solve-plex}(k_1)\) transitions the state to \((4, 7)\), completing one full iteration. Repeating this process leads to termination at \((3, 6)\), where the corresponding vertical coordinate \(6\) represents the optimum solution \(s^*=6\).

\noindent \underline{\textbf{Example of the pseudo LB technique.}}
\begin{figure}[t]
    \centering
    \subfigure[An example graph.]{
        \includegraphics[height=0.16\textwidth,width=0.256\textwidth, trim=6cm 7.5cm 15.2cm 4cm, clip]{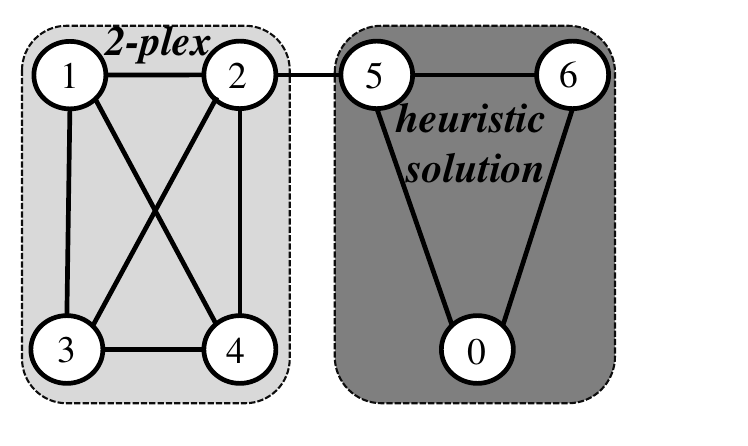} 
        \label{fig:improved_framework_example}
    }
    \subfigure[Illustration of the example.]{
    \includegraphics[height=0.16\textwidth,width=0.16\textwidth, trim=6cm 22.5cm 9cm 1cm, clip]{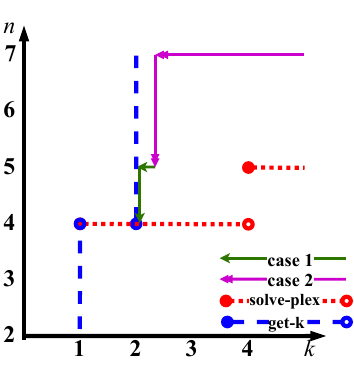} 
        \label{fig:improved_framework_example2}
    }
    \vspace{-0.2in}
    \caption{An example of Algorithm~\ref{alg:improved-iterative-search} with $\gamma = 0.75$.}
    \label{Fig:improved_iterative_example}
\end{figure}
Figure~\ref{fig:improved_framework_example} shows a graph with 7 vertices and Figure~\ref{fig:improved_framework_example2} presents the corresponding iterative process.
For clarity, line segments near $k=2$ in Figure~\ref{fig:improved_framework_example2} are slightly shifted along the positive $k$-axis; this is solely for visualization and does not indicate non-integer $k$ values. Assume that the heuristic solution obtained in each iteration corresponds to the subgraph in the dark-colored box, i.e., $G[\{0,5,6\}]$, in Figure~\ref{fig:improved_framework_example}, where \(lb\mbox{-}plex = 3\) according to Line 1 of Algorithm~\ref{alg:kplex-search}.
As shown in Figure~\ref{fig:improved_framework_example2}, using the pseudo LB technique to solve the problem with \(\gamma = 0.75\) in Figure~\ref{fig:improved_framework_example} requires two iterations. The first iteration, represented by the double-arrowed purple solid line in the figure, corresponds to the second case discussed in the proof of Lemma~\ref{lem:connection}. At this stage, \(pseudo\mbox{-}lb = \lfloor(3 + 7)/2\rfloor = 5 > \overline{s}_1 = 4\). Consequently, the branch-and-bound search in Line 4 of Algorithm~\ref{alg:kplex-search} fails to find a solution, resulting in \(S = \emptyset\). Thus, \(s_1 = \max\{5, 0\} = 5\), which initiates the second iteration.
In the second iteration, \(pseudo\mbox{-}lb = \lfloor(5 + 3)/2\rfloor = 4 \leq \overline{s}_2 = 4\), which corresponds to the first case discussed in the proof of Lemma~\ref{lem:connection}. Since the pruning during the branch-and-bound search uses \(pseudo\mbox{-}lb = 3\), it does not affect the computation of the solution \(S\). Thus, \(s_2 = \max\{3, 4\} = 4\), satisfying the condition in Line 4 of Algorithm 3. The final solution, i.e., $G[\{1,2,3,4\}]$, is in Figure~\ref{fig:improved_framework_example} within the light-colored box.

\section{Additional Experiments}\label{sec:more-exp}
\noindent \underline{\textbf{Number of solved instances.}} 
Figures~\ref{Fig:dimacs10_various_gamma_appendix} and~\ref{Fig:realworld_various_gamma_appendix} illustrate the trends in the number of solved instances over time for \texttt{IterQC}, \texttt{FastQC}, and \texttt{DDA} on the 10th DIMACS and real-world datasets, under $\gamma$ values of 0.55, 0.6, 0.7, 0.8, and 0.9. From these results, it is easy to see that under the same time constraint, for any given value of $\gamma$ or dataset, \texttt{IterQC} consistently solves more instances compared to the baseline algorithms. Remarkably, regardless of the dataset or $\gamma$, \texttt{IterQC} solves more instances within 3 seconds than the larger number achieved by either baseline algorithm within 3 hours. For instance, at $\gamma = 0.6$, \texttt{IterQC} solves 53 instances in only 3 seconds, while \texttt{FastQC} and \texttt{DDA} solve only 41 and 58 instances, respectively, even after 3 hours of computation. Across all tested values of $\gamma$, \texttt{IterQC} achieves improved practical performance. On the 10th DIMACS dataset, it solves at least 79 out of 84 instances, while on the real-world dataset, it solves at least 132 out of 139 instances.

\begin{figure*}[t]
    \centering
    \subfigure[$\gamma = 0.55$]{
        \includegraphics[width=0.18\textwidth]{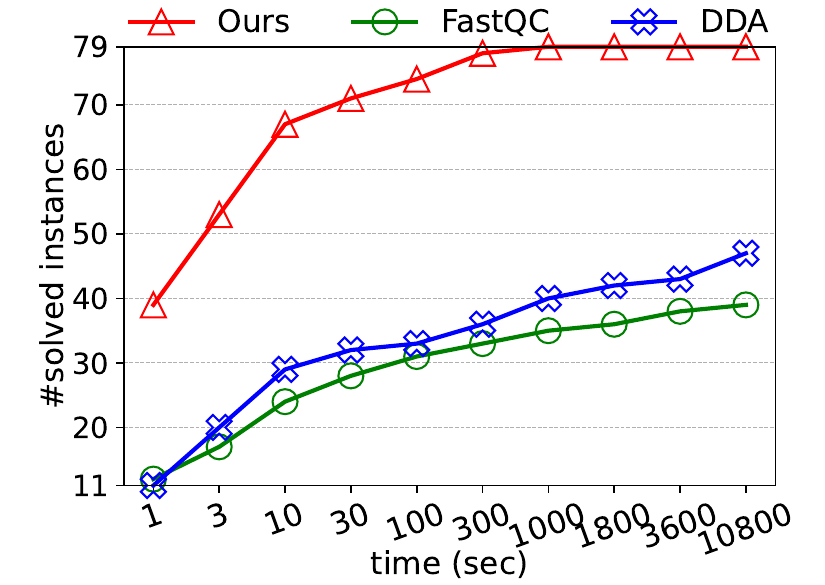}
        \label{fig:dimacs055}
    }
    \subfigure[$\gamma = 0.6$]{
        \includegraphics[width=0.18\textwidth]{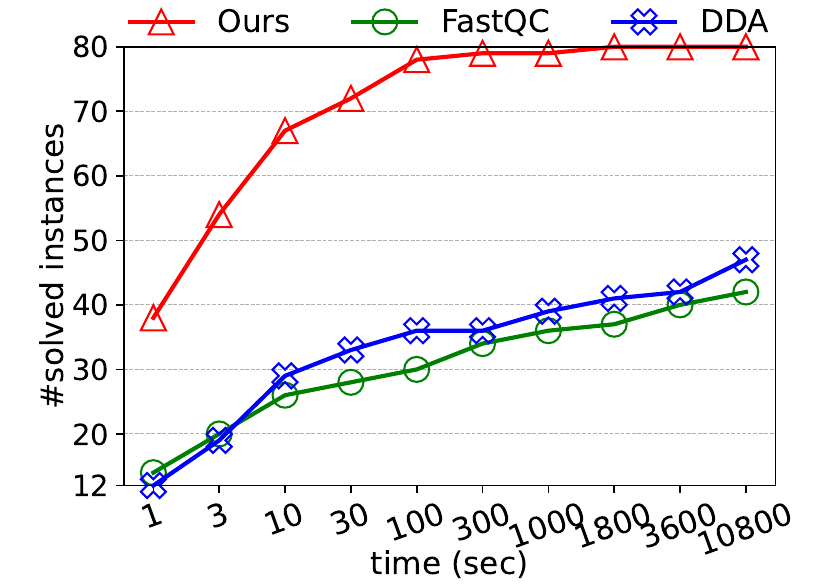}
        \label{fig:dimacs06}
    }
    \subfigure[$\gamma = 0.7$]{
        \includegraphics[width=0.18\textwidth]{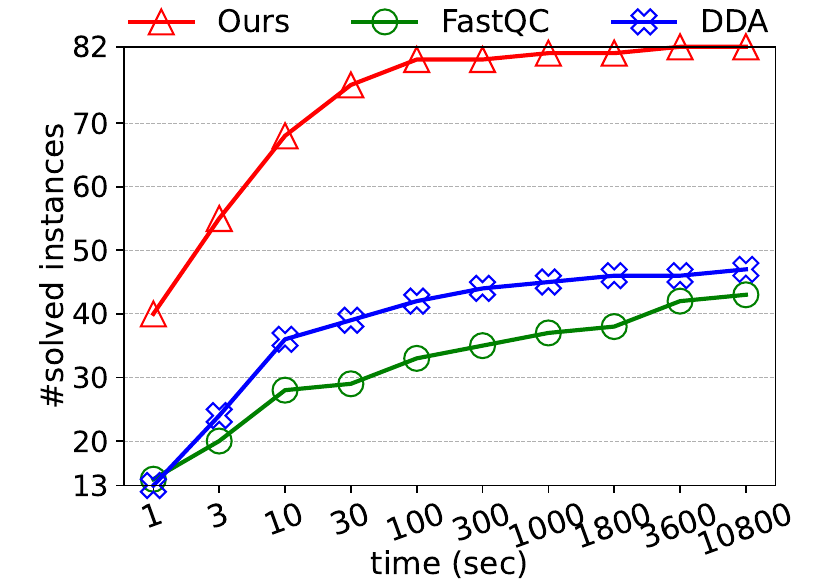}
        \label{fig:dimacs07}
    }
    \subfigure[$\gamma = 0.8$]{
        \includegraphics[width=0.18\textwidth]{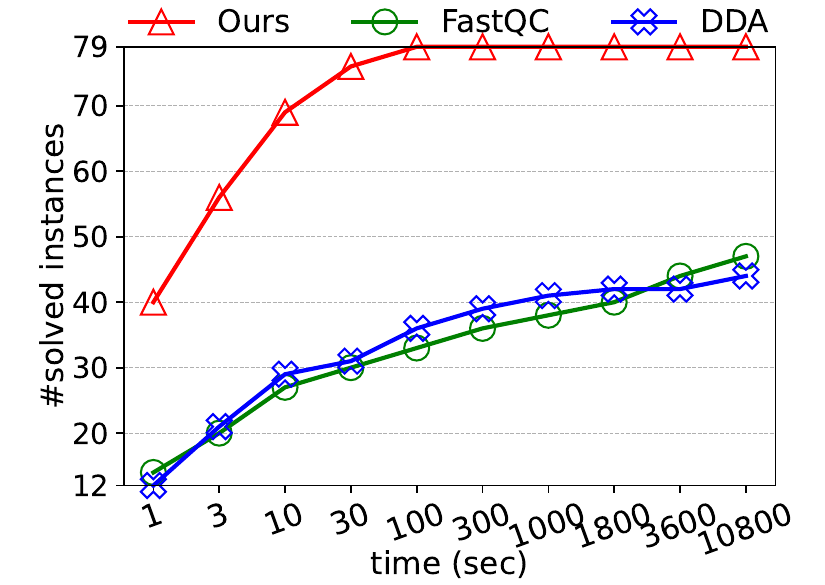}
        \label{fig:dimacs08}
    }
    \subfigure[$\gamma = 0.9$]{
        \includegraphics[width=0.18\textwidth]{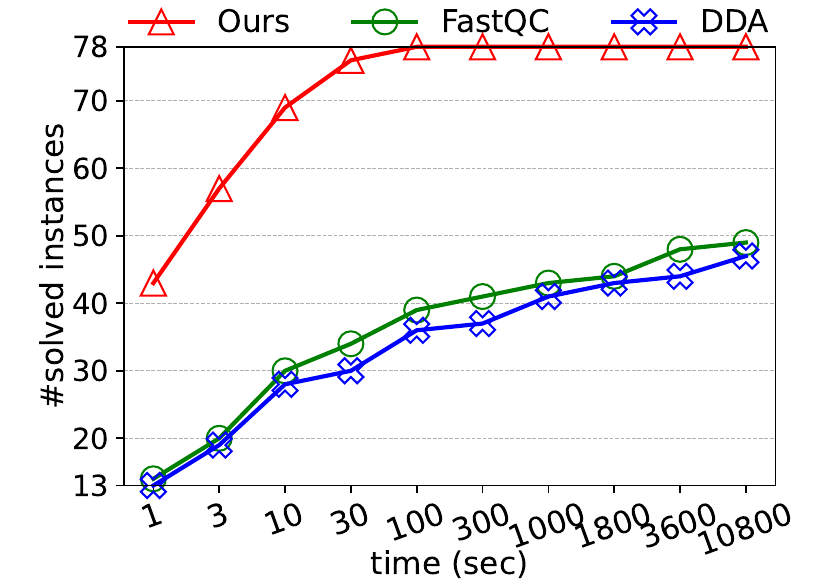}
        \label{fig:dimacs09}
    }
    \vspace{-0.2in}
    \caption{Number of solved instances on 10th DIMACS graphs.}
    \label{Fig:dimacs10_various_gamma_appendix}
\end{figure*}

\begin{figure*}[t]
    \centering
    \subfigure[$\gamma = 0.55$]{
        \includegraphics[width=0.18\textwidth]{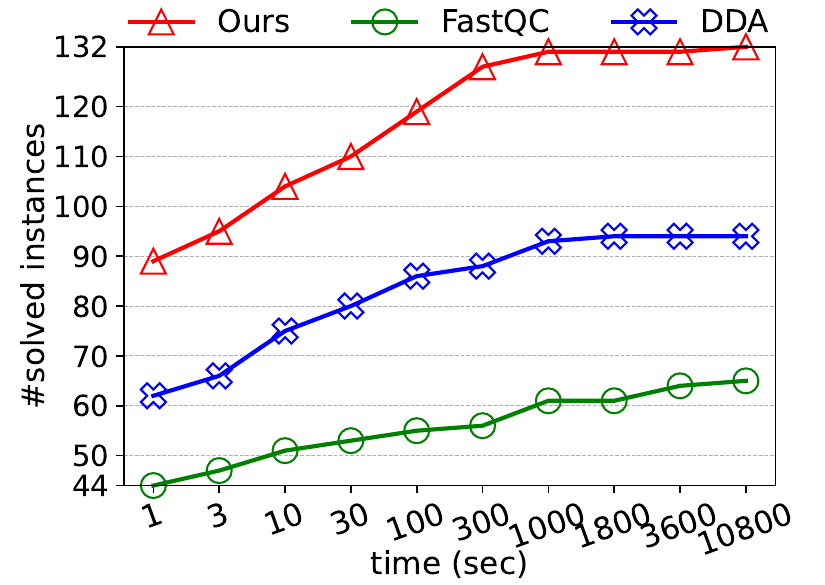}
        \label{fig:realworld055}
    }
    \subfigure[$\gamma = 0.6$]{
        \includegraphics[width=0.18\textwidth]{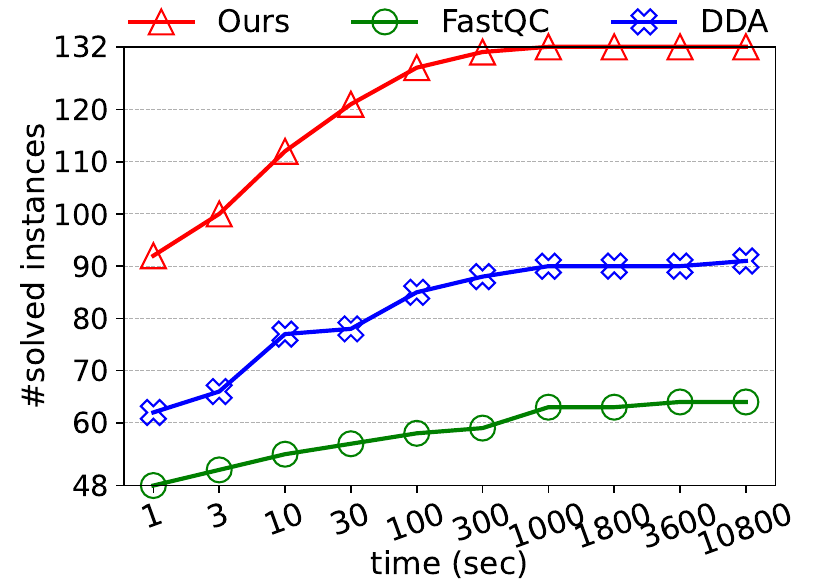}
        \label{fig:realworld06}
    }
    \subfigure[$\gamma = 0.7$]{
        \includegraphics[width=0.18\textwidth]{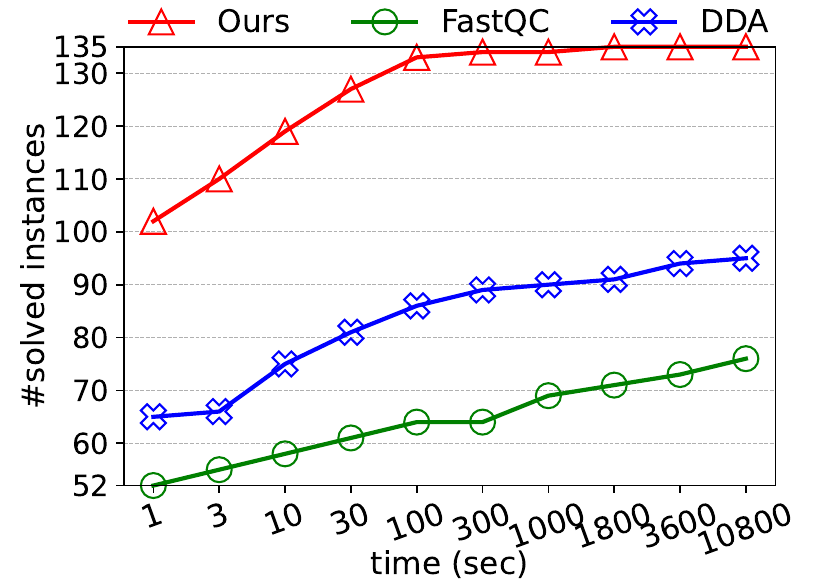}
        \label{fig:realworld07}
    }
    \subfigure[$\gamma = 0.8$]{
        \includegraphics[width=0.18\textwidth]{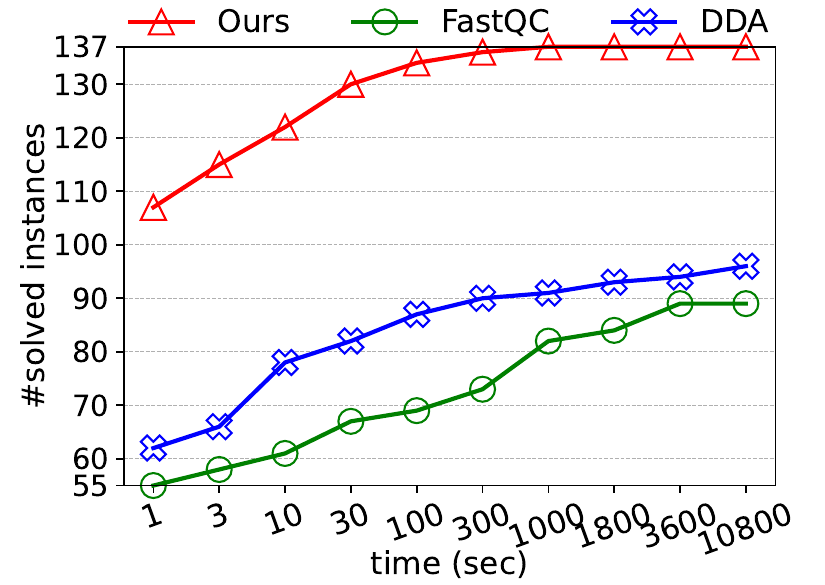}
        \label{fig:realworld05}
    }
    \subfigure[$\gamma = 0.9$]{
        \includegraphics[width=0.18\textwidth]{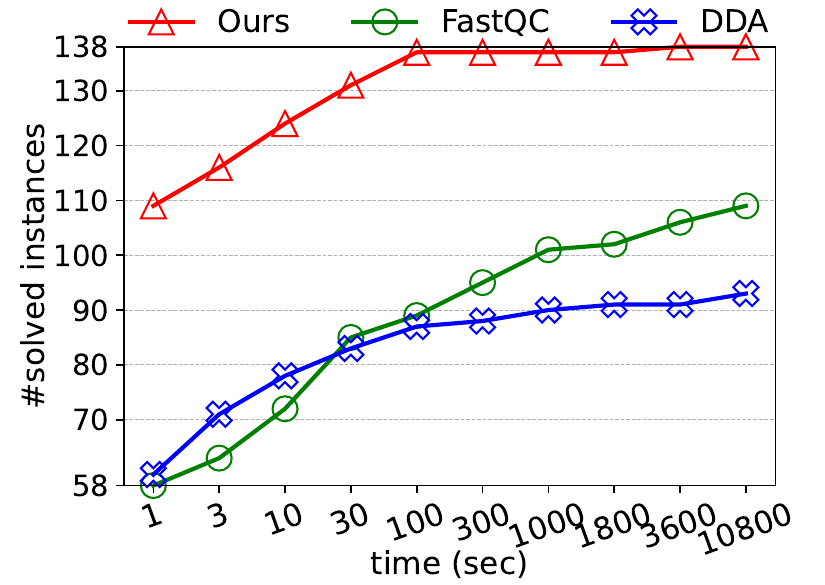}
        \label{fig:realworld09}
    }
    \vspace{-0.2in}
    \caption{Number of solved instances on real-world graphs.}
    \label{Fig:realworld_various_gamma_appendix}
\end{figure*}

\noindent \underline {\textbf{Performance on representative instances.}} Tables~\ref{table:representitive_65} and~\ref{table:representitive_85} summarize the performance of \texttt{IterQC} and the two baseline algorithms across 30 representative instances under $\gamma$ values of 0.65 and 0.85. The results demonstrate that \texttt{IterQC} consistently outperforms the baseline algorithms in terms of solving capability within the given time limit of 3 hours. In particular, \texttt{IterQC} successfully solves all instances except for G29 under $\gamma = 0.65$, achieving nearly complete coverage. In contrast, the two baseline algorithms, \texttt{FastQC} and \texttt{DDA}, exhibit a significant number of timeouts under all values of $\gamma$. For instance, at $\gamma = 0.65$, \texttt{FastQC} and \texttt{DDA} fail to solve 25 and 18 instances, respectively, out of the 30 cases. These results highlight the robustness and efficiency of \texttt{IterQC} in handling challenging instances, further validating its scalability and superiority over the baseline methods.
On the other hand, we can observe that in most cases, as \( \gamma \) increases, the runtime of \texttt{FastQC} gradually decreases, as seen in instances like G4 and G5, or the difference remains small, as in G13, G14, and G17. However, for the \texttt{IterQC} algorithm, the trend may not always be monotonic, as shown in instances like G15 and G16. This is because \texttt{IterQC} starts iterating from an upper bound, and its performance is influenced by both the core number and \( \gamma \) in a heuristic manner. If the initial bound is already close to the final solution, the algorithm can converge in just a few iterations. As a result, \( \gamma \) not only affects the complexity of the search space but also impacts the number of iterations in \texttt{IterQC}, making its runtime not necessarily monotonic with respect to \( \gamma \).

\begin{table}[t]
        \centering
        \captionsetup{font=small}
        \caption{Runtime performance (in seconds) of \texttt{IterQC}, \texttt{FastQC}, and \texttt{DDA} on 30 instances with $\gamma = 0.65$.}
        \label{table:representitive_65}
        \vspace{-0.15in}
        \scalebox{0.8}{
        \begin{tabular}{l|l l l|l|l l l}
        \hline
         \textbf{ID} & \texttt{IterQC} & \texttt{FastQC} & \texttt{DDA} & \textbf{ID} & \texttt{IterQC} & \texttt{FastQC} & \texttt{DDA} \\ \hline
        G1 & \textbf{0.008} & OOT & 0.25 & G16 & \textbf{6.49} & OOT & OOT \\ 
        G2 & 6.19 & OOT & \textbf{1.04} & G17 & \textbf{0.04} & 1041.55 & 3.22\\ 
        G3 & \textbf{0.55} & OOT & OOT & G18 & OOT & OOT & OOT \\ 
        G4 & \textbf{0.03} & 2211.78 & 101.39 & G19 & \textbf{13.80} & OOT & 26.85 \\ 
        G5 & \textbf{6.07} & OOT & OOT & G20 & \textbf{2.42} & OOT & 13.39 \\ 
        G6 & OOT & OOT & OOT & G21 & \textbf{6.47} & OOT & 21.64 \\ 
        G7 & \textbf{0.33} & OOT & OOT & G22 & \textbf{9.70} & OOT & 6845.3 \\ 
        G8 & \textbf{5.33} & OOT & OOT & G23 & \textbf{0.30} & OOT & 26.90 \\ 
        G9 & \textbf{0.01} & 3.89 & 1.92 & G24 & \textbf{4.82} & OOT & OOT \\ 
        G10 & \textbf{4.08} & OOT & OOT & G25 & \textbf{1.80} & OOT & 91.48 \\ 
        G11 & \textbf{135.1} & OOT & OOT & G26 & \textbf{0.49} & OOT & OOT \\ 
        G12 & \textbf{3.38} & OOT & OOT & G27 & \textbf{26.13} & OOT & OOT \\ 
        G13 & \textbf{5.87} & 33.99 & OOT & G28 & \textbf{19.65} & OOT & 40.70 \\ 
        G14 & \textbf{5.01} & 98.83 & OOT & G29 & \textbf{36.31} & OOT & OOT \\ 
        G15 & \textbf{3.06} & OOT & OOT & G30 & \textbf{23.28} & OOT & OOT \\ \hline
        \end{tabular}}
\end{table}
\begin{table}[t]
    \centering
    \captionsetup{font=small}
    \caption{Runtime performance (in seconds) of \texttt{IterQC}, \texttt{FastQC}, and \texttt{DDA} on 30 instances with $\gamma = 0.85$.}
    \label{table:representitive_85}
    \vspace{-0.15in}
    \scalebox{0.8}{
    \begin{tabular}{l|l l l|l|l l l}
        \hline
        ID & \texttt{IterQC} & \texttt{FastQC} & \texttt{DDA} & ID & \texttt{IterQC} & \texttt{FastQC} & \texttt{DDA} \\ \hline
        G1 & \textbf{0.08} & OOT & 3.43 & G16 & \textbf{6.49} & OOT & OOT \\
        G2 & 5.02 & OOT & \textbf{0.88} & G17 & \textbf{0.015} & 1115.24 & 517.41 \\ 
        G3 & \textbf{0.30} & OOT & OOT & G18 & \textbf{91.45} & OOT & OOT \\ 
        G4 & \textbf{0.02} & 1.18 & 21.53 & G19 & \textbf{12.32} & OOT & 16.63 \\ 
        G5 & \textbf{0.50} & 85.27 & OOT & G20 & \textbf{2.65} & OOT & OOT \\ 
        G6 & \textbf{5.01}& 5.05 & OOT & G21 & \textbf{13.03} & OOT & 21.88 \\ 
        G7 & \textbf{0.21} & 203.62 & OOT & G22 & \textbf{18.28} & OOT & OOT \\ 
        G8 & \textbf{5.13} & OOT & OOT & G23 & \textbf{0.24} & OOT & 25.93 \\ 
        G9 & \textbf{0.13} & 3.90 & 0.48 & G24 & \textbf{8.84} & OOT & OOT \\ 
        G10 & \textbf{4.14} & OOT & OOT & G25 & \textbf{1.90} & OOT & 46.32 \\ 
        G11 & \textbf{7.95} & OOT & OOT & G26 & \textbf{0.38} & OOT & OOT \\ 
        G12 & \textbf{0.45} & OOT & OOT & G27 & \textbf{29.11} & OOT & OOT \\ 
        G13 & \textbf{5.90} & 34.95 & 2490.8 & G28 & \textbf{17.61} & OOT & 39.78 \\ 
        G14 & \textbf{3.48} & 102.83 & OOT & G29 & \textbf{0.99} & OOT & 93.58 \\ 
        G15 & \textbf{0.19} & OOT & OOT & G30 & \textbf{21.31} & OOT & OOT \\ \hline
    \end{tabular}}
\end{table}

\noindent \underline {\textbf{Ablation studies.}} Tables~\ref{table:representitive_65_ablation_appendix} to~\ref{table:representitive_85_ablation_appendix} present the results of the ablation studies conducted on 30 representative instances, demonstrating the impact of preprocessing and the pseudo lower bound (pseudo LB) techniques on practical performance. The results clearly indicate significant improvements across all values of $\gamma$. For the preprocessing technique, the optimization effect is clear in nearly every case. For example, at $\gamma = 0.65$, instance G29 achieves a speedup of approximately $ 20.79\times$ compared to the setting without the preprocessing technique. Regarding the pseudo LB technique, substantial benefits are observed in most instances, particularly for smaller values of $\gamma$. At $\gamma = 0.65$, 23 out of 30 instances show accelerated running time. Specifically, for G11 at $\gamma = 0.65$ and G15 at $\gamma = 0.85$, the introduction of the pseudo LB technique transforms previously unsolvable instances within the 3-hour limit into solvable ones. For G15 at $\gamma = 0.85$, the running time is reduced to 57.88 seconds, corresponding to a speedup of $186\times$. These findings highlight the critical contributions of both preprocessing and pseudo LB techniques in improving efficiency of the proposed method.
\begin{table}[t]
\captionsetup{font=small}
    \caption{Runtime performance (in seconds) of \texttt{IterQC}, \texttt{IterQC-PP}, and \texttt{IterQC-PLB} on 30 instances with $\gamma = 0.65$.}
    \label{table:representitive_65_ablation_appendix}
    \vspace{-0.15in}
    \scalebox{0.8}{
    \begin{tabular}{l|l l l|l|l l l}
    \hline
        ID & \texttt{IterQC} & \texttt{-PP} & \texttt{-PLB} & ID & \texttt{IterQC} & \texttt{-PP} & \texttt{-PLB} \\ \hline
        G1 & 0.01& 0.1& \textbf{0.01} & G16 & \textbf{5.87} & 102.0& 6.52\\
        G2 & 0.01& 0.07& \textbf{0.01} & G17 & 0.3& 4.95& \textbf{0.28} \\
        G3 & \textbf{0.55} & 2.11& 0.55& G18 & OOT & OOT & OOT \\
        G4 & \textbf{0.04} & 0.18& 0.05& G19 & \textbf{4.08} & 29.05& 4.33\\
        G5 & \textbf{6.07} & 10.14& 14.32& G20 & \textbf{1.8} & 32.42& 2.71\\
        G6 & \textbf{6.19} & 7.53& 6.47& G21 & \textbf{2.42} & 73.74& 2.95\\
        G7 & 0.33& 1.57& \textbf{0.23} & G22 & \textbf{4.82} & 77.83& 6.55\\
        G8 & \textbf{3.06} & 4.07& 12.26& G23 & \textbf{5.33} & 90.88& 5.93\\
        G9 & \textbf{0.02} & 0.23& 0.03& G24 & \textbf{6.47} & 79.45& 7.78\\
        G10 & 36.31& \textbf{35.66} & 180.13& G25 & 19.65& 144.75& \textbf{19.2} \\
        G11 & \textbf{135.12} & 501.39& OOT & G26 & \textbf{26.13} & 159.5& 29.94\\
        G12 & \textbf{3.38} & 3.64& 7.35& G27 & \textbf{9.7} & 116.42& 11.26\\
        G13 & 0.49& 1.74& \textbf{0.46} & G28 & \textbf{6.49} & 134.9& 6.57\\
        G14 & \textbf{5.01} & 9.74& 5.87& G29 & \textbf{13.8} & 368.1& 13.91\\
        G15 & OOT & OOT & OOT & G30 & \textbf{23.28} & 460.12& 27.91\\ \hline
    \end{tabular}}
\end{table}

\begin{table}[t]
    \centering
    \captionsetup{font=small}
    \caption{Runtime performance (in seconds) of \texttt{IterQC}, \texttt{IterQC-PP}, and \texttt{IterQC-PLB} on 30 instances with $\gamma = 0.85$.}
    \label{table:representitive_85_ablation_appendix}
    \vspace{-0.15in}
    \scalebox{0.8}{
    \begin{tabular}{l|l l l|l|l l l}
    \hline
        ID & \texttt{IterQC} &\texttt{-PP} & \texttt{-PLB} & ID & \texttt{IterQC} & \texttt{-PP} & \texttt{-PLB} \\ \hline
        G1 & \textbf{0.08} & 0.14& 0.1& G16 & \textbf{5.9} & 17.33& 6.34\\
        G2 & 0.13& 0.46& \textbf{0.13} & G17 & \textbf{0.24} & 1.67& 0.26\\
        G3 & 0.3& 0.86& \textbf{0.28} & G18 & 91.45& 5334.55& \textbf{74.72} \\
        G4 & \textbf{0.01} & 0.08& 0.02& G19 & \textbf{4.14} & 12.31& 5.96\\
        G5 & 0.5& 3.15& \textbf{0.38} & G20 & \textbf{1.9} & 11.16& 2.05\\
        G6 & 5.02& 7.39& \textbf{2.83} & G21 & \textbf{2.65} & 25.22& 3.31\\
        G7 & 0.21& 0.85& \textbf{0.14} & G22 & \textbf{8.84} & 33.29& 12.81\\
        G8 & 0.19& 0.81& \textbf{0.17} & G23 & \textbf{5.13} & 35.33& 6.9\\
        G9 & 0.02& 0.09& \textbf{0.02} & G24 & \textbf{13.03} & 45.33& 13.29\\
        G10 & \textbf{0.99} & 1.64& 3.51& G25 & \textbf{17.61} & 59.02& 18.97\\
        G11 & 7.95& 92.42& \textbf{7.42} & G26 & \textbf{29.11} & 82.92& 31.36\\
        G12 & 0.45& 1.12& \textbf{0.4} & G27 & \textbf{18.28} & 69.84& 19.63\\
        G13 & \textbf{0.38} & 0.94& 0.44& G28 & \textbf{6.49} & 58.82& 6.95\\
        G14 & \textbf{3.48} & 7.3& 4.13& G29 & \textbf{12.32} & 86.57& 13.46\\
        G15 & \textbf{57.88} & 107.74& OOT & G30 & \textbf{21.31} & 163.58& 25.44\\ \hline
    \end{tabular}}
\end{table}

\noindent \underline {\textbf{Performance on artificial datasets.}} We test our algorithms on synthetic graphs: Scale-Free (SF) graphs and Small-World (SW) graphs~\cite{rahman2024pseudo}. SF graphs are generated by the Barab\'{a}si–Albert model with parameter $w$, starting from a star graph of $w$ vertices and incrementally adding new vertices, each connected to $w$ existing vertices, until the graph contains $n$ vertices. SW graphs are generated by the Watts–Strogatz model with rewiring probability $p$. It begins with a cycle graph of $n$ vertices, where each vertex is connected to its $d$ nearest neighbors, forming a ring lattice. Then, for each edge $(u, v)$, the endpoint $v$ is randomly rewired to a non-adjacent vertex with probability $p$. In general, the graphs become denser as $w$ or $d$ increases. We fix $n = 10^6$, set $p = 0.2, \gamma=0.75$, and vary $w$ and $d$ in $\{10, 15, 20, 25, 30\}$. The time limit is set to 3 hours (10,800 seconds). We report the running times in the Table~\ref{table:artificial-datasets}.

We observe that our \texttt{IterQC} algorithm consistently outperforms all baseline methods. The running time of all algorithms increases with larger $w$ or $d$, as higher values of these parameters lead to denser graphs. Moreover, the performance gains of \texttt{IterQC} become more pronounced as $w$ increases in SF graphs and $d$ decreases in SW graphs. This is because a larger $w$ in SF graphs results in denser local structures, where the preprocessing step can effectively reduce the graph size. In contrast, while increasing $d$ in SW graphs adds more edges, the uniform structural properties limit the pruning effectiveness, resulting in less speedup.

\begin{table}[t]
    \centering
    \captionsetup{font=small}
    \caption{Runtime performance (in seconds) of \texttt{IterQC}, \texttt{FastQC}, and \texttt{DDA} on artificial datasets with $\gamma = 0.75$.}
    \label{table:artificial-datasets}
    \vspace{-0.15in}
    \scalebox{0.8}{
    \begin{tabular}{l|l l l|l|l l l}
    \hline
        \multicolumn{4}{c|}{\textbf{SF}} & \multicolumn{4}{c}{\textbf{SW}} \\
\hline
        $w$ & \texttt{IterQC} & \texttt{FastQC} & \texttt{DDA} & $d$ & \texttt{IterQC} & \texttt{FastQC} & \texttt{DDA} \\ \hline
        10 & 0.14 & 9.59 & OOT & 10 & 0.18 & 4.44 & OOT \\ 
        15 & 0.30 & 18.56 & OOT & 15 & 0.26 & 4.66 & OOT \\ 
        20 & 0.52 & 32.66 & OOT & 20 & 0.35 & 5.39 & OOT \\ 
        25 & 0.63 & 57.33 & 1.06 & 25 & 0.88 & 6.05 & OOT \\ 
        30 & 0.93 & 146.74 & OOT & 30 & 1.83 & 7.29 & OOT \\ \hline
    \end{tabular}}
\end{table}

\noindent \underline {\textbf{Memory consumption.}} We evaluate the memory usage of \texttt{IterQC} and two baseline algorithms across 30 representative graphs, as reported in Table~\ref{table:memory}. Memory usage is measured as the peak memory consumption during execution. For baseline methods that encountered timeouts, we record the maximum memory usage observed prior to termination.
We observe that the memory usage is comparable among all algorithms. In addition, we can prove that \texttt{IterQC} has the space complexity of $O(m+n)$, which is the same as the baselines and is linear with respect to the size of the input graph. The experimental results are aligned with the theoretical analysis.

\begin{table}[h]
    \centering
    \captionsetup{font=small}
    \caption{Memory consumption (in GB) of \texttt{IterQC}, \texttt{FastQC}, and \texttt{DDA} on 30 instances with $\gamma = 0.75$.}
    \label{table:memory}
    \vspace{-0.15in}
    \scalebox{0.8}{
    \begin{tabular}{l|l l l|l|l l l}
    \hline
        ID & \texttt{IterQC} & \texttt{FastQC} & \texttt{DDA} & ID & \texttt{IterQC} & \texttt{FastQC} & \texttt{DDA} \\ \hline
        G1 & \textbf{0.0} & \textbf{0.0}& 0.01& G16 & 1.73& 6.06& \textbf{0.92} \\  
        G2 & \textbf{0.0} & \textbf{0.0}& 0.02& G17 & \textbf{0.11} & 0.2& 0.12\\  
        G3 & \textbf{0.03} & 0.09& 0.13& G18 & 3.63& 11.56& \textbf{2.46} \\  
        G4 & \textbf{0.01} & \textbf{0.01}& \textbf{0.01}& G19 & 1.08& 16.14& \textbf{0.48} \\  
        G5 & \textbf{0.05} & 0.15& 0.09& G20 & \textbf{0.57} & 2.06& 0.59\\  
        G6 & 0.1& \textbf{0.07} & 0.12& G21 & \textbf{0.88} & 3.19& 0.95\\  
        G7 & \textbf{0.03} & 0.1& 0.04& G22 & \textbf{1.11} & 3.69& 16.17\\  
        G8 & \textbf{0.02} & 0.06& 0.09& G23 & 1.13& 3.93& \textbf{1.09} \\  
        G9 & \textbf{0.01} & \textbf{0.01}& 0.02& G24 & \textbf{1.46} & 4.89& 21.1\\  
        G10 & \textbf{0.02} & 0.05& 0.06& G25 & 4.66& 16.14& \textbf{1.88} \\  
        G11 & \textbf{0.09} & 0.25& 0.16& G26 & \textbf{2.45} & 7.54& 32.57\\  
        G12 & \textbf{0.04} & 0.13& 0.13& G27 & \textbf{1.95} & 6.48& 29.26\\  
        G13 & \textbf{0.04} & 0.06& 0.6& G28 & \textbf{1.87} & 6.69& 1.91\\  
        G14 & \textbf{0.17} & 0.21& 0.8& G29 & \textbf{3.75} & 13.56& 3.98\\  
        G15 & \textbf{0.22} & 0.74& 0.45& G30 & \textbf{5.21} & 18.99& 9.17\\ \hline 
    \end{tabular}}
\end{table}

\end{document}